\documentclass[12pt]{article}
\usepackage{epsf}
\usepackage{hyperref}
\usepackage{graphicx, rotating}
\usepackage{subfigure}
\usepackage[margin=2.5cm]{geometry}
\usepackage{amsmath}
\usepackage{footnote}
\usepackage[bottom]{footmisc}
\usepackage{threeparttable}
\usepackage{caption}
\usepackage{amsthm}
\usepackage{enumitem}
\usepackage{graphics}
\usepackage{float}
\usepackage{amssymb}
\usepackage{natbib}
\usepackage{lscape}
\usepackage{array}
\usepackage{setspace}
\usepackage{verbatim}
\usepackage{multirow}

\newtheorem{theorem}{Theorem}
\newtheorem{lemma}{Lemma}

\newtheorem{definition}{Definition}

\bibpunct{(}{)}{;}{a}{,}{,}

\begin{document}
\title{\LARGE Survival trees for left-truncated and right-censored data, with application to time-varying covariate data}
\author{Wei Fu \qquad Jeffrey S. Simonoff \\\\ New York University}
\date{\today}
\maketitle

\begin{abstract}
Tree methods (recursive partitioning) are a popular class of nonparametric methods for analyzing data. One extension of the basic tree methodology is the survival tree, which applies recursive partitioning to censored survival data. There are several existing survival tree methods in the literature, which are mainly designed for right-censored data. We propose two new survival trees for left-truncated and right-censored (LTRC) data, which can be seen as a generalization of the traditional survival tree for right-censored data. Further, we show that such trees can be used to analyze survival data with time-varying covariates, essentially building a time-varying covariates survival tree. Implementation of the methods is easy, and simulations and real data analysis results show that the proposed methods work well for LTRC data and survival data with time-varying covariates, respectively.
\end{abstract}

\section{Introduction}
Right-censored data are often studied using a (semi-)parametric model such as the Cox proportional hazards model. However, the parametric assumptions imposed by these models are often either not met or unrealistic in practice. Therefore, more flexible non-parametric models are desired. Survival trees and survival forests are among the most popular non-parametric alternatives to the Cox model. 

Various authors have proposed tree methods for right-censored data. The key feature that distinguishes different tree methods is the splitting criterion. In the literature, most survival tree algorithms employ similarity/dissimilarity measures of the survival profile for splitting. \cite{segal1988regression} extended the regression tree to right-censored data by replacing the conventional splitting criteria with rules based on two sample statistics. Traditionally, splitting criteria are geared to optimizing within-node homogeneity. In contrast, Segal's algorithm rewards splits that result in large between-node separation (note there is no algebraic equivalence between those two approaches in general).

For the proportional hazards model, the hazard function at time $t$ for an individual with covariates $\mathbf{x}$ is 
\[ \lambda(t|\mathbf{x}) = \lambda_0(t)s(\mathbf{x})	 \]
where $s(\mathbf{x}) \geq 0$ and $\lambda_0(t)$ is the baseline hazard. \cite{1992relative} proposed a method to construct a tree representing the relative risk function, $s(\mathbf{x})$. Their algorithm splits the covariate space based on a rule that maximizes the reduction in the one-step deviance realized by the split, which is defined as the difference between the log-likelihood of the saturated model and the maximized log-likelihood. The baseline cumulative hazard function is estimated by the Nelson-Aalen \citep{nelson1972theory,aalen1978nonparametric} estimator. The algorithm is implemented in the \texttt{R} package \texttt{rpart}. 

\cite{Hothorn06} (hereafter HHZ) implemented an unbiased survival tree using the log-rank test as the splitting method. They successfully embedded the survival tree algorithm into a large framework of conditional inference trees, which has the desirable property of selecting the splitting variable in an unbiased way (an unbiased tree has the property that when there is no relationship between the response and any predictors all predictors have the same probability of being the split variable). Specifically, the algorithm transfers the bivariate survival outcome (response) $(T, \delta_i)$ ($\delta_i = 0$ if $T$ is censoring time; otherwise $1$) into a univariate outcome called the log-rank score and proceeds with it as the response variable. It can be shown that splitting based on the univariate log-rank score is equivalent to splitting based on the bivariate survival outcome $(T,\delta_i)$ using the log-rank test. Another distinct feature of the algorithm of HHZ is that it does not prune. It stops splitting when the algorithm determines there is no need to split further. The algorithm is implemented in the \texttt{R} package \texttt{partykit}.

\subsection{Survival trees with left-truncation data and time-varying covariates}
All of these algorithms deal with the most basic setup of survival outcome -- right-censored data with time-independent covariates. However, other types of survival data such as left-truncated and right-censored (LTRC) data and survival data with time-varying covariates arise commonly in practice. According to \cite{klein2003survival}, left-truncation mainly occurs under two situations: when the event time $T$ is the age of the subject and persons are not observed from birth but rather from some other time $V$ corresponding to their entry into the study; and when $T$ is measured from some landmark, but only subjects who experience some intermediate event at time $V$ are included in the study. Ideally, we want a survival tree algorithm that can handle LTRC and time-varying covariates survival data. However, it turns out time-varying covariates are difficult to deal with using tree methods.

Trees recursively partition the sample space by asking the question ``Is $X_j < C$?''. Observations with answer ``yes'' go to one sub-node and those with answer ``no'' go to the other sub-node. If $X_j$ is a time-independent covariate, then the partition is well-defined for every observation for every possible cut point $C$. The situation is different when $X_j(t)$ is a time-varying covariate. For a specific observation $i$, it is possible that $X_{ij}(t) < C$ for $t < t^*$ but $X_{ij}(t) \geq C$ for $t \geq t^*$, so it is not clear to which sub-node observation $i$ should go. 

\cite{bacchetti1995survival} proposed the first method to extend the tree method to handle time-varying covariates. Their algorithm handles time-varying covariates by splitting each observation into several pseudo-subjects based on the split $x(t)\leq C$, where each pseudo-subject represents a non-overlapping time interval and either $x(t)> C$ or $x(t) \leq C$ in the entire interval. For observation $i$, such a procedure splits the observation at time $t^*$ into two pseudo-subjects, one with $X_{j}(t) < C$ (since $t<t^*$) and one with $X_{j}(t) \geq C$ (since $t \geq t^*$). These two pseudo-subjects can then go to separate sub-nodes. The algorithm uses the log-rank test as the splitting criterion. Since some pseudo-subjects are LTRC data by construction (since $t \geq t^*$), the splitting method is based on a log-rank test \citep{mantel1966evaluation} that is adjusted to accommodate LTRC data. 

\cite{bertolet2012tree} proposed partitioning the data based on time-varying Cox models with time-varying indicators $I_{x(t) \leq C}$ as regression variables. This is equivalent to the approach of \cite{bacchetti1995survival} with 
the log-rank test as the splitting criterion. Other proposed methods includes \cite{huang1998piecewise}, \cite{xu2002survival} and \cite{wallace2014time}. Unfortunately, none of these proposed methods are implemented in publicly available software.

There are few survival tree algorithms that have been proposed to specifically handle left-truncated and right-censored (LTRC) data. One exception is the method proposed by \cite{bacchetti1995survival}, which accommodates LTRC data as a middle step in order to split on time-varying covariates. However, like other existing survival tree methods for time-varying covariates, the algorithm is not publicly available. In addition, existing survival tree methods for 
time-varying covariates are not based on established survival tree algorithms. Such inconsistency might make them less likely to be adopted by users, which may help explain why they are not widely used. Rather, in practice, people usually choose to use the Cox proportional hazard model to handle LTRC data and/or survival data with time-varying covariates.

We propose two new survival tree methods for LTRC data. These two tree methods are simple extensions based on two widely-used survival tree methods, and are easy to implement in practice. Through data reformulation, the proposed LTRC survival tree methods can then be used to fit survival data with time-varying covariates, which makes each of the methods as versatile as the Cox model, being applicable to right-censored data, LTRC data and survival data with time-varying covariates.

In Section $2$, we lay out the details of the two proposed tree methods for LTRC data; Section $3$ investigates the properties and performance of the proposed tree methods through simulations; Section $4$ shows the application of the proposed methods on a real data set; Section $5$ gives the details about how to fit time-varying covariates survival trees using the proposed LTRC trees; Section $6$ discusses the properties of the time-varying covariates survival trees; Section $7$ contains two real data applications of time-varying covariates survival trees and Section $8$ summarizes conclusions.

\section{New LTRC trees by extending existing survival tree methods}
In this section, we provide two examples of generalizing existing survival tree methods to handle left-truncated and right-censored (LTRC) data. To be specific, we extend two widely used survival tree algorithms from the \texttt{R} packages \texttt{partykit} and \texttt{rpart}, respectively, which implement the survival tree algorithms of HHZ and \cite{1992relative}, respectively. 

\subsection{Extending the survival tree of HHZ}
\subsubsection{Log-rank score for right-censored data}

The conditional inference tree of HHZ measures the association of $Y$ and a predictor $X_j$ by linear statistics of the form 
\[
T_j(L_n,w) = vec \left(\sum_{i=1}^n w_ig_i(X_{ji})h \left( Y_i,(Y_i,...Y_n)\right)^T\right)\in \mathbb{R}^{p_jq}                            
\]
(equation $3.1$ in HHZ), where $g_j : \mathcal{X}_j \to \mathbb{R}^{p_j} $ is a nonrandom transformation of covariate $X_j$ and $h : \mathcal{Y} \times \mathcal{Y}^n \to \mathbb{R}^q$ is the influence function of the response $Y$. For a univariate numeric response $Y$, the choice of influence function is the identity, i.e.
$h \left( Y_i,(Y_i,...Y_n)\right) = Y_i	$.
For right-censored data, subjects can be represented as a triple $(t_i,\delta_i, \mathbf{x}_i)$, $i=1, 2,..., n$, where $t_i$ is the observed event time or censored time for the $i$th subject, $\delta_i = 1$ if $t_i$ is the event time and  $\delta_i = 0$ if $t_i$ is the censored time and $\mathbf{x}_i$ is the covariate vector for the $i$th subject.  We also assume that censoring is noninformative given $\mathbf{x}_i$. Then, the response variable for the
$i$th subject is $Y_i = (t_i, \delta_i)$. The influence function for such a bivariate response is the so-called log-rank score
\[	h \left( Y_i,(Y_i,...Y_n)\right) = U_i = \delta_i - \sum_{j=1}^{r_i(t)} \frac{\delta_j}{n-r_j(t)+1}\text{,}	\]
where $r_j(t) = \sum_{i=1}^n I_{\{t_i \leq t_j\}}$ is the number of observations who died or were censored before or at time $t_j$ \cite[equation~(13) ]{hothorn2003exact}. The main function of the log-rank score is to assign a univariate value $U_i$ (scalar) to the bivariate response $Y_i = (t_i, \delta_i)$, so the algorithm can then execute in the same way as in the univariate numeric response case. 

The log-rank score was first proposed by \cite{peto1972asymptotically}, who derived general (asymptotically efficient) rank invariant test procedures for detecting differences between two groups of independent observations.
Let $\hat{S}(t)$ denote the empirical survival curve. For a censored observation, for which the true event time (unobserved) is known to lie in an interval over which $\hat{S}(t)$ drops from $a$ to $b$, or for an observed event time $t_i$, where $\hat{S}(t)$ drops from $a$ to $b$, the log-rank score is

\[ U = \frac{a \log (a)-b \log (b)}{a-b} \text{.}	\] 

$U$ is approximately $1 + \log \hat{S}(t)$ for an exact observed event time $t$, and is exactly $\log \hat{S}(t)$ for a censored survival time $t$. Let $\theta_A$ denote the parameters of survival distribution of group A, which consists of $m$ independent observations, and let $\theta_B$ denote the parameters of group B with $N-m$ independent observations. Then the test of $H_0: \theta_A = \theta_B$ versus $H_a: \theta_A \neq \theta_B$ can be constructed from the test statistics $T_A = \sum^m_{i=1} U_i$, which is the sum of scores from group A. Under $H_0$, $T_A$ has the distribution of a sum of $m$ random variables chosen randomly from $U_1,...,U_N$.  

 \cite{peto1972asymptotically} also established that under $H_0$, the null hypothesis that groups have the same distribution, using the statistics $\sum_{i \in j} U_i$ and $(O_j-\hat{E}_j)$ to describe the $j$th group are equivalent, which means that using the log-rank score and log-rank test to compare survival curves of different groups are equivalent, under the condition of independent observations. 

\subsubsection{Log-rank score for LTRC data}

Let the triple $(L_i, R_i, \delta_i)$ denote the $i$th LTRC observation, where $L_i$ is the left-truncation time, $R_i$ is the observed survival time/censored time and $\delta_i$ is the event indicator. The goal is to construct an influence function that can map the triple $(L_i, R_i, \delta_i)$ into a scalar $U_i$ (i.e. a log-rank score for LTRC data), which is equivalent to testing $H_0$ using either $\sum_{i \in j} U_i$ or the log-rank test. 

\cite{pan1998rank} extended the rank invariant tests of \cite{peto1972asymptotically} to left-truncated and interval-censored data. The log-rank score for left-truncated and interval-censored data in \cite{pan1998rank} is
\[	U_i	= \frac{\hat{S}(l_i) \log \hat{S}(l_i) - \hat{S}(r_i) \log \hat{S}(r_i)}{\hat{S}(l_i)-\hat{S}(r_i)}-\log \hat{S}(L_i)\]
Here, $L_i$ is the left-truncation time and $(l_i, r_i)$ is the interval in which the true event time lies. The log-rank score for LTRC data can be derived from this score equation as a special case. 

For LTRC data, if an observation is censored at time $t_i$, then we only know that the true event time lies in the interval $(t_i, \infty)$; if it is observed at time $t_i$, that means it lies in the interval $\lim_{\Delta \to 0}(t_i-\Delta, t_i+\Delta) $. Through simple calculation and the fact that $\hat{S}(\infty) = 0$, it is easy to determine the log-rank score for our LTRC observation $(L_i, R_i, \delta_i)$ as
\begin{equation}
	U_i = 1 + \log \hat{S}(R_i) - \log \hat{S}(L_i)	\hspace{0.1in} \text{if} \hspace{0.1in} \delta_i=1  \end{equation}
and 
\begin{equation}
U_i = \log \hat{S}(R_i)	- \log \hat{S}(L_i)\hspace{0.1in} \text{if} \hspace{0.1in} \delta_i=0, 
\end{equation}
where $L_i$ is the left-truncation time and $R_i$ is the event/right-censoring time. Note that $\hat{S}$ is the nonparametric maximum likelihood estimator (NPMLE) of the survival function. In practice, such an estimator can be constructed using the product-limit estimator, i.e. Kaplan-Meier (KM) estimator, by redefining the risk set. Note the interpretation of such a product-limit estimator is now conditional, because only for time $t \geq \tau$, $\tau = \min \{L_i: i = 1,...,n\}$, can the nonparametric estimator be calculated and is consistent \citep{gross1996nonparametric,tsai1987note}. Since in $(1)$ and $(2)$ only the ratio $\hat{S}(L_i)/\hat{S}(R_i)$ matters, whether $\hat{S}$ is a conditional estimator or an unconditional estimator is immaterial.   

We will refer to this extended LTRC tree as LTRCIT (LTRC tree based on Conditional Inference Tree).
\subsection{Extending the survival tree of \cite{1992relative}}

\cite{1992relative} proposed a survival tree algorithm based on the assumption of proportional hazards. Specifically, let $(t,\delta, \mathbf{x})$ denote an observation where $t$ is the observed event/censored time, $\delta$ is the event indication and $\mathbf{x}$ is the vector of covariates. The sample consists of $n$ independent observations $(t_i,\delta_i, \mathbf{x}_i), i=1,2,...,n$. Then, the full likelihood of the proportional hazards model
\[ \lambda(t|\mathbf{x}) = \lambda_0(t)s(\mathbf{x})	 \]
of the sample for a tree $T$ can be expressed as
\[ L = \Pi_{h \in \widetilde{T}} \Pi_{i \in S_h}  \lambda_h(t_i)^{\delta_i}e^{-\Lambda_h(t_i)},\]
where $\widetilde{T}$ is the set of terminal nodes (leaves), $S_h$ is the set of observation labels , $\{i: \mathbf{x}_i \in \chi_h\}$ for observations in the region $\chi_h$ corresponding to node $h$, and $ \lambda_h(t)$ and $\Lambda_h(t_i)$ are the hazard and cumulative hazard function for node $h$, respectively. Assume the proportional hazards model

\[	\lambda_h(t) = \theta_h \lambda_0(t)	\] 
is true, where $\theta_h $ is the nonnegative relative risk of node $h$ and $\lambda_0(t)$ is the baseline hazard. Then the full likelihood can be written as

\[L = \Pi_{h \in \widetilde{T}} \Pi_{i \in S_h}  (\theta_h \lambda_0(t_i))^{\delta_i}e^{-\Lambda_0(t_i) \theta_h}\text{,}\] 
where $\Lambda_0(t_i)$ is the baseline cumulative hazard function. \cite{1992relative} estimate $\theta_h$ by 

\[	\hat{\theta}_h	 = \frac{\sum_{i \in S_h} \delta_i}{\sum_{i \in S_h} \Lambda_0(t_i)}\text{,}\]
where $\Lambda_0$ is estimated using all of the data at the root node by the Nelson-Aalen estimator. This can be seen to be the observed number of events divided by the expected number of events in node $h$ assuming observations in node $h$ are randomly sampled from the root node. The deviance for node $h$ is 

\[ R(h) = 2[L_h(\text{saturated}) - L_h(\hat{\theta}_h)]	\]
where $L_h(\text{saturated}) $ is the log-likelihood of the saturated model and $L_h(\hat{\theta}_h)$ is the maximized log-likelihood when $\Lambda_0(t)$ is known. The splitting criterion is the reduction of the node deviance residual

\[	D_{\text{parent}} - \{ D_{\text{left daughter node}} + D_{\text{right daughter node}}\}	\] 
where $D_h = \sum_{i \in h} d_i $, with the contribution of the $i$th observation being 

\begin{equation}
d_i = 2\left[ \delta_i \log \left(\frac{\delta_i}{\Lambda_0(t_i)\hat{\theta}_h} \right) -\left(\delta_i - \Lambda_0(t_i)\hat{\theta}_h \right) \right] \text{.}
\end{equation}
	
An equivalent approach is based on Poisson regression. Let $\lambda$ be an event rate, $t_i$ be exposure time for observation $i$ and $c_i$ is the observed event count for observation $i$. Then the within node deviance residual for a Poisson regression tree is 

\[	D =  \sum \left[ c_i \log \left(\frac{c_i}{\hat{\lambda} t_i} \right) -\left(c_i - \hat{\lambda} t_i \right) \right]	\]
with $\hat{\lambda} = \frac{\sum c_i}{\sum t_i}$. Comparing this to the node deviance residual of a survival tree, one can easily see that they are equivalent if we replace $c_i$ by $\delta_i$ and  $t_i$ by $\Lambda_0(t_i)$.

This is how the survival tree is fit in \texttt{rpart}. That is, the algorithm first estimates the baseline hazard 
$\Lambda_0(t_i)$ based on the entire training data and then fits a Poisson regression tree by treating $\Lambda_0(t_i)$ as the new $t_i$ and treating $\delta_i$ as the new $c_i$.

\subsubsection{Equivalent Poisson regression tree for LTRC data}

The full log-likelihood for right-censored data $(t_i,\delta_i, \mathbf{x}_i), i=1,2,...,n$ is
\begin{equation}
 \log L  = \sum^n_{i=1} \left[ \delta_i \log \lambda(t_i) - \Lambda(t_i) \right]  = \sum^n_{i=1}  \left[ \delta_i \log \lambda(t_i) - \int^{t_i}_0 \lambda(\mu)d \mu \right],
\end{equation}
while the log-likelihood for left-truncated and right-censored (LTRC) data $(L_i, R_i, \delta_i, \mathbf{x}_i), i=1,2,...,n$ is

\begin{equation}
\log L  = \sum^n_{i=1} \left[ \delta_i \log \lambda(R_i) - \left(\Lambda(R_i) - \Lambda\left(L_i\right)\right) \right]  = \sum^n_{i=1}  \left[ \delta_i \log \lambda(R_i) - \int^{R_i}_{L_i} \lambda(\mu)d \mu \right] \text{,}
\end{equation}
where $L_i$ is the left-truncation time and $R_i$ is the right-censored time for observation $i$. Since the only difference is the replacement of $\Lambda(t_i)$ with $\Lambda(R_i) - \Lambda(L_i)$, replacing $\Lambda_0(t_i)$ in \cite{1992relative} with $\Lambda_0(R_i) - \Lambda_0(L_i)$ effectively extends the model to LTRC data.

Three steps are needed to implement the method. First, estimation of the cumulative function $\Lambda_0(t)$ is still based on all of the LTRC data. Note that observation $(L_i, R_i, \delta_i, \mathbf{x}_i)$ is only counted as in the risk set for time $t$ when $L_i \leq t \leq R_i$. Next, the ``exposure time'' for observation $i$ is computed using $\hat{\Lambda}_0(R_i) - \hat{\Lambda}_0(L_i)$ based on the estimated cumulative function $\hat{\Lambda}_0(t)$. Finally, we fit a Poisson regression tree by treating the calculated $\hat{\Lambda}_0(R_i) - \hat{\Lambda}_0(L_i)$ as the new $t_i$ and treating $\delta_i$ as the new $c_i$. The extended LTRC tree is called LTRCART (LTRC tree based on CART framework).

\section{Properties of the two LTRC trees}

The survival tree of HHZ is unbiased in terms of selecting the splitting variable, which means it selects each covariate with equal probability of splitting under the condition of independence between response and covariates. This suggests that the extended LTRC tree based on it (i.e. LTRCIT) is also unbiased, while that based on \cite{1992relative} (i.e. LTRCART) is not. Simulation results show that this is indeed the case. The details of the unbiasedness test can be found in the supplemental material, available at \url{http://people.stern.nyu.edu/jsimonof/survivaltree}.

\subsection{Recovering the correct tree structure \label{Sec 3.1}}
We first explore the proposed trees' ability to recover the correct underlying tree structure of the data. The simulation setup is as follows.
 
The left truncation time $L$ is generated as independent uniform $[0, U]$, with $U$ taking on values from $\{1,2,3\}$ to represent different truncation rates. There are $6$ covariates $X_1,...,X_6$, where $X_1, X_4$ randomly take values from the set $\{1, 2, 3, 4, 5\}$, $X_2, X_5$ are binary$\{1, 2\}$ and $X_3,X_6$ are $U[0, 2]$. Only the first three covariates $X_1, X_2, X_3$ determine the distribution of survival (event) time $T$. The survival time $T$ has distribution according to the values of $X_1, X_2, X_3$ by the structure given in Figure \ref{fig1}
\begin{figure}[H]
  \begin{center}
   \captionsetup{justification=centering}
      \includegraphics[width=6 in,height =3 in]{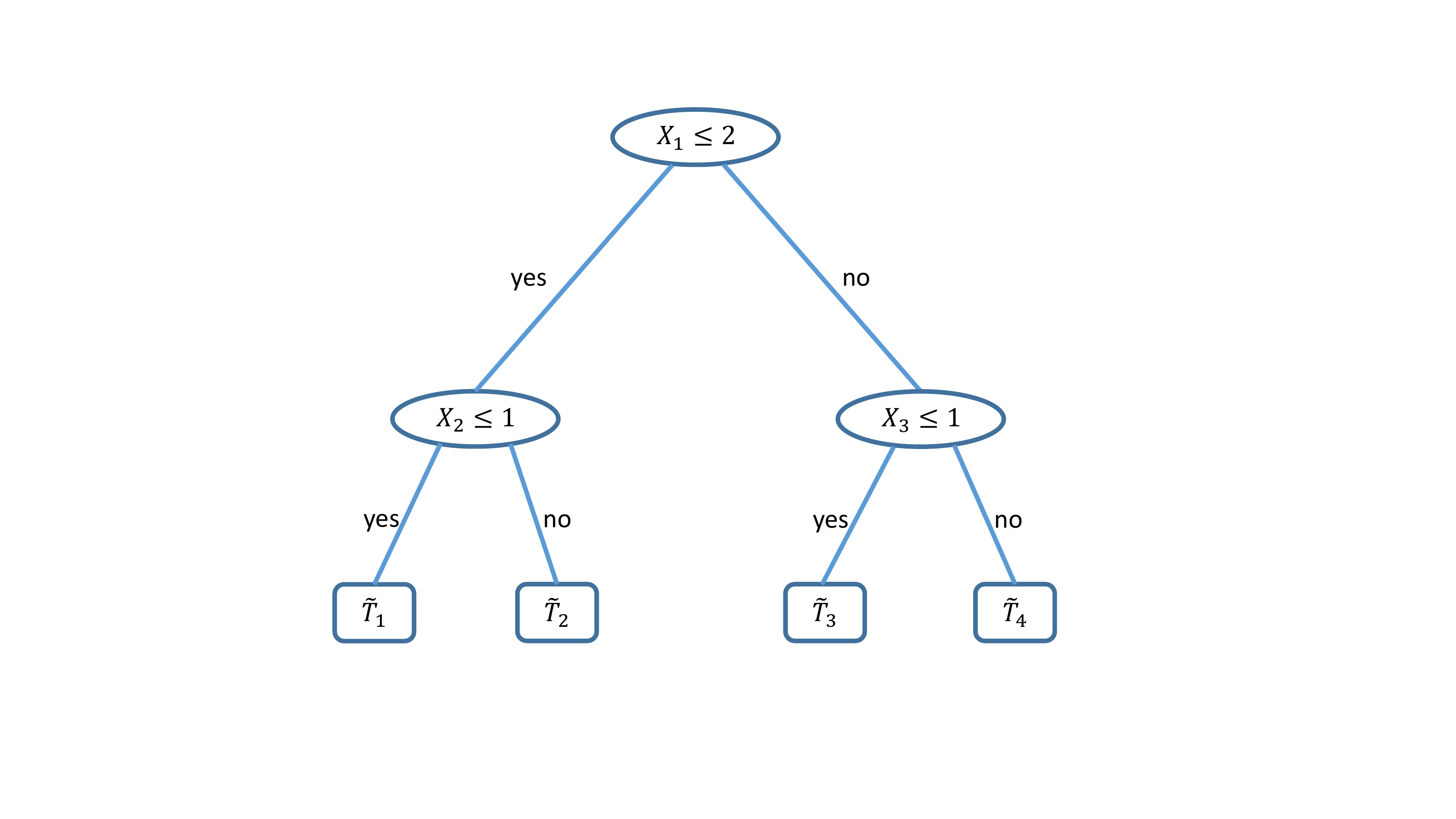}
   \end{center}
     \caption{\label{fig1} Tree structure used in simulations of Section \ref{Sec 3.1}} 
 \end{figure} 
 \noindent
If the generated $T < L$, i.e. the survival time is less than the left-truncation time, then this observation is discarded. Otherwise, the observation is retained, with censoring time $C = L + D$, where $D$ has an exponential distribution. If $C < T$, then this observation is censored ($\delta = 0$), otherwise the survival time $T$ is observed ($\delta=1$). Note here $D$ and $L$ are both independently generated from $T$ and from each other. \\

\noindent
We generate $T$ from $5$ different distributions:
\begin{itemize}
	\item Exponential with four different values of $\lambda$ from $\{0.1, 0.23, 0.4, 0.9\}$.
	\item Weibull distribution with shape parameter $\alpha = 0.9$, which corresponds to decreasing hazard with time. The scale parameter $\beta$ takes the values $\{7.0, 3.0, 2.5, 1.0\}$.
	\item Weibull distribution with shape parameter $\alpha = 3$, which corresponds to increasing hazard with time. The scale parameter $\beta$ takes the values $\{2.0, 4.3, 6.2, 10.0\}$.
	\item Log-normal distribution with location parameter $\mu$ and scale parameter $\sigma$ with $4$ different pairs $(\mu, \sigma) = \{(2.0, 0.3), (1.7, 0.2), (1.3, 0.3), (0.5, 0.5)\}$.
	\item Bathtub-shaped hazard model \citep{hjorth1980reliability}. The survival function is given by
	\[ S(t;a,b,c) = \frac{\exp(-\frac{1}{2}at^2)}{(1+ct)^{b/c}} \]
	with $b=1$, $c=5$ and $a$ set to take value $\{0.01, 0.05, 0.1, 0.7\}$.
\end{itemize}
Note that for the exponential distribution and two Weibull distributions, the proportional hazards assumption is satisfied for the four groups. Each of the five distributions has two possible censoring rates: light censoring with about $20\%$ observations being censored and heavy censoring with about $50\%$ observations being censored. Coupled with three different truncation rates, each distribution has six different combinations of censoring rate and truncation rate. Parameters in each distribution are selected to assure that the pattern of mean values of $T$ across nodes is similar across different distributions.  Figure \ref{density} shows the density $f_T$ in each leaf for different distributions.   

\begin{figure}[t]
  \begin{center}
   \captionsetup{justification=centering}
      \includegraphics[width=6in,height =4in]{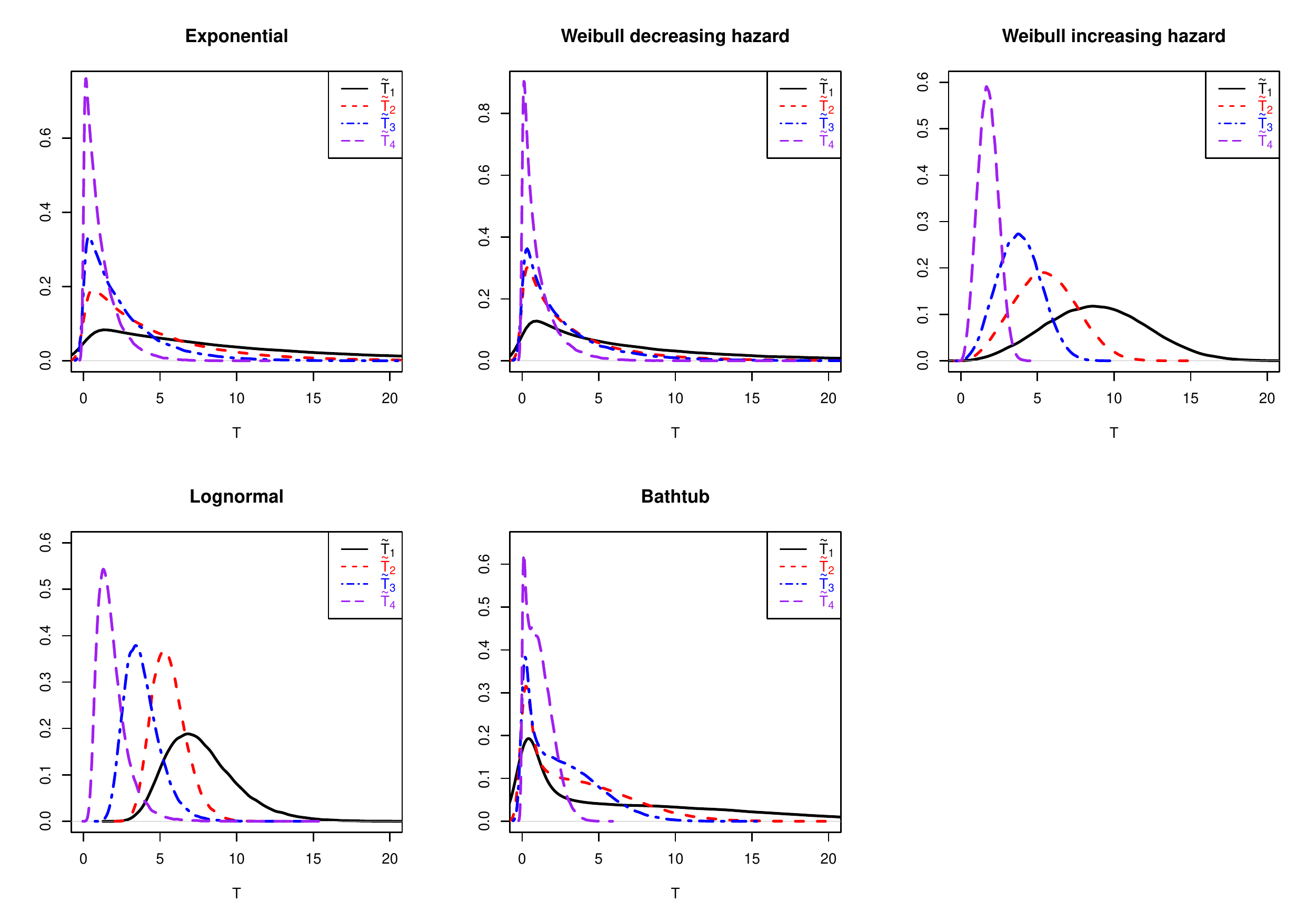}
   \end{center}
   \caption{\label{density} Density plots of survival time $T$ for each terminal node by distribution type.} 
 \end{figure}

\begin{table}[t]
\center
\caption{\label{table2} Tree structure recovery rate in percentages.} 
\scalebox{0.77}{%
\hspace*{-2cm}
\begin{tabular}{ ccccccccccccccccc}
\hline
\multicolumn{2}{c}{N=100} && \multicolumn{2}{c}{Exponential} &&  \multicolumn{2}{c}{Weibull-I}  &&  \multicolumn{2}{c}{Weibull-D}  &&  \multicolumn{2}{c}{Lognormal}  &&  \multicolumn{2}{c}{Bathtub} \\  \cline{4-5}  \cline{7-8} 
\cline{10-11} \cline{13-14} \cline{16-17}
 Censor.rate & Truncation && \footnotesize{LTRCIT} & \footnotesize{LTRCART} &&  \footnotesize{LTRCIT} &\footnotesize{LTRCART} && \footnotesize{LTRCIT} & \footnotesize{LTRCART} &&  \footnotesize{LTRCIT} & \footnotesize{LTRCART}  && \footnotesize{LTRCIT} & \footnotesize{LTRCART}\\ 
\hline
 Light & $U[0, 1]$ &&  $6.0$&$6.8$  &&  $58.9$&$32.6$   &&   $2.2$&$3.9$  &&   $47.6$&$15.7$ &&   $8.2$&$5.3$ \\ 
 Heavy & $U[0, 1]$ &&  $2.1$&$2.6$  &&  $17.6$&$9.5$    &&   $0.4$&$1.6$   &&   $10.6$&$2.9$  &&   $1.7$&$1.8$ \\ 
 Light & $U[0, 2]$ &&  $4.1$&$6.7$  &&  $61.7$&$35.6$   &&   $1.8$&$4.6$  &&   $45.0$&$20.2$ &&   $12.1$&$10.7$ \\ 
 Heavy & $U[0, 2]$ &&  $1.0$&$2.8$  &&  $23.2$&$12.1$   &&   $0.6$&$1.0$   &&   $14.8$&$5.7$  &&   $2.4$&$3.3$ \\ 
 Light & $U[0, 3]$ &&  $2.9$&$4.7$  &&  $57.8$&$41.8$   &&   $1.2$&$3.8$  &&   $36.3$&$23.4$ &&   $13.4$&$12.0$ \\ 
 Heavy & $U[0, 3]$ &&  $1.2$&$1.9$  &&  $27.2$&$14.5$   &&   $0.4$&$0.9$    &&   $13.8$&$7.3$  &&   $3.2$&$3.3$ \\  \hline
 
\multicolumn{2}{c}{N=300} && \multicolumn{2}{c}{Exponential} &&  \multicolumn{2}{c}{Weibull-I}  &&  \multicolumn{2}{c}{Weibull-D}  &&  \multicolumn{2}{c}{Lognormal}  &&  \multicolumn{2}{c}{Bathtub} \\  \cline{4-5}  \cline{7-8} 
\cline{10-11} \cline{13-14} \cline{16-17}
 Censor.rate & Truncation && \footnotesize{LTRCIT} & \footnotesize{LTRCART} &&  \footnotesize{LTRCIT} &\footnotesize{LTRCART} && \footnotesize{LTRCIT} & \footnotesize{LTRCART} &&  \footnotesize{LTRCIT} & \footnotesize{LTRCART}  && \footnotesize{LTRCIT} & \footnotesize{LTRCART}\\ 
\hline
 Light & $U[0, 1]$ &&  $69.1$&$47.6$  &&  $83.8$&$85.1$   &&   $62.2$&$40.7$   &&   $84.8$&$49.7$  &&   $61.1$&$26.1$ \\ 
 Heavy & $U[0, 1]$ &&  $41.4$&$18.8$  &&  $71.0$&$37.9$   &&   $35.4$&$18.4$   &&   $66.5$&$10.7$  &&   $15.2$&$5.2$ \\ 
 Light & $U[0, 2]$ &&  $65.0$&$50.3$  &&  $84.2$&$85.8$   &&   $57.8$&$44.7$   &&   $84.4$&$60.5$  &&   $75.3$&$48.7$ \\ 
 Heavy & $U[0, 2]$ &&  $38.7$&$22.9$  &&  $76.7$&$47.2$   &&   $32.4$&$18.4$   &&   $75.8$&$19.8$  &&   $28.4$&$9.5$ \\ 
 Light & $U[0, 3]$ &&  $57.0$&$48.0$  &&  $82.3$&$85.5$   &&   $51.4$&$43.2$   &&   $86.0$&$66.7$  &&   $81.6$&$62.2$ \\ 
 Heavy & $U[0, 3]$ &&  $34.3$&$20.5$  &&  $80.7$&$59.9$   &&   $28.0$&$18.0$   &&   $81.9$&$31.6$  &&   $46.1$&$18.7$ \\  \hline
 
\multicolumn{2}{c}{N=500} && \multicolumn{2}{c}{Exponential} &&  \multicolumn{2}{c}{Weibull-I}  &&  \multicolumn{2}{c}{Weibull-D}  &&  \multicolumn{2}{c}{Lognormal}  &&  \multicolumn{2}{c}{Bathtub} \\  \cline{4-5}  \cline{7-8} 
\cline{10-11} \cline{13-14} \cline{16-17}
 Censor.rate & Truncation && \footnotesize{LTRCIT} & \footnotesize{LTRCART} &&  \footnotesize{LTRCIT} &\footnotesize{LTRCART} && \footnotesize{LTRCIT} & \footnotesize{LTRCART} &&  \footnotesize{LTRCIT} & \footnotesize{LTRCART}  && \footnotesize{LTRCIT} & \footnotesize{LTRCART}\\ 
\hline
 Light & $U[0, 1]$ &&  $81.5$&$77.8$  &&  $84.6$&$91.6$  &&  $82.2$&$76.0$   &&   $85.3$&$74.9$  &&   $78.8$&$53.3$ \\ 
 Heavy & $U[0, 1]$ &&  $72.2$&$45.7$  &&  $84.3$&$69.8$  &&  $68.0$&$40.7$   &&   $82.6$&$27.2$  &&   $29.8$&$9.5$ \\ 
 Light & $U[0, 2]$ &&  $80.7$&$82.9$  &&  $86.3$&$91.6$  &&  $80.9$&$78.6$   &&   $85.9$&$79.3$  &&   $82.7$&$76.9$ \\ 
 Heavy & $U[0, 2]$ &&  $70.8$&$50.8$  &&  $84.8$&$77.9$  &&  $67.1$&$44.8$   &&   $83.9$&$39.5$  &&   $54.4$&$23.4$ \\ 
 Light & $U[0, 3]$ &&  $78.5$&$78.4$  &&  $82.7$&$91.8$  &&  $79.2$&$75.2$   &&   $84.6$&$83.3$  &&   $84.1$&$86.1$ \\ 
 Heavy & $U[0, 3]$ &&  $69.4$&$52.2$  &&  $81.5$&$86.9$  &&  $65.3$&$44.5$   &&   $85.4$&$59.0$  &&   $68.8$&$42.0$ \\  \hline
\end{tabular}
\hspace*{-2cm}
}
\end{table}

We run 1,000 simulation trials for each setting to see how well the two proposed LTRC trees recover the correct tree structure. Table \ref{table2} gives the percentage of the time the correct tree structure is found for each setting. 

The LTRCART tree requires a pruning strategy. We prune the fitted Poisson tree back by selecting the subtree with smallest ten-fold cross-validation error. The usual $1$-SE rule (select the smallest tree with cross-validation error less than one standard error above the minimized value) appears to be too pessimistic for this splitting criterion. A performance comparison of the default $0$-SE rule vs. the usual $1$-SE for LTRCART in terms of recovering the correct tree structure can be found in the supplemental material, from which one can see that the default choice clearly outperforms the usual $1$-SE rule.

To understand the results, it might be helpful to look at Figure \ref{density}. From the density plot, it seems that the easiest distributions to distinguish among the four groups (terminal nodes) are the Weibull with increasing hazard and the lognormal distribution, as we can see those groups are distinguishable even when left-truncation and right-censoring are considered. This is exactly what the results in Table \ref{table2} show, as both LTRC tree algorithms perform better for these distributions, especially when the sample size is small.

The censoring rate has an obvious impact on the trees as we can see that heavy censoring reduces the recovery rate in all cases. The impact of heavy censoring is large when the sample size is smaller, presumably because larger samples bring stability to the trees, which partially offsets the effects caused by information lost due to censoring. From Table \ref{table2}, we note that the left-truncation distribution has minimal impact on the results. Since all of the  distributions are heavy tailed, their ``tail'' distributions ($f_T$ with large $T$) are more decisive than the ``head'' distributions ($f_T$ with small $T$) in terms of uniqueness of each distribution. In other words, the important characteristics of each distribution are more represented in the ``tail'' than in the ``head.'' Since left-truncation causes information loss in the ``head'' while the right-censoring causes information loss in the ``tail,'' censoring has more impact than left-truncation in Table \ref{table2}.

As a general phenomenon, increasing the sample size helps both LTRC trees to successfully recover the correct tree structure, but LTRCART seems to benefit more from a large sample size than does LTRCIT. In other words, LTRCART is more sensitive to sample size than LTRCIT. LTRCIT generally outperforms LTRCART, even in the proportional hazards cases, such as the Exponential and the two Weibull distributions. It also seems to be less sensitive to high censoring compared to LTRCART. 

\subsection{Prediction performance}

\subsubsection{Performance measure}
To compare different methods, we first need a performance measure. Unlike in a classification or regression problem where misclassification rate or MSE is the obvious choice, there is no single obvious way of measuring a model's prediction power for survival data. 
   
The most popular measure of error in the survival context is the Brier score, along with its integrated version introduced by \cite{graf1999assessment}. 
For right-censored survival data, let $(X_i, Y_i,\delta_i)$ denote the observed information of the $i$th observation, where $X_i$ is a vector of covariates' value, $Y_i$ is the observed survival time and $\delta_i$ is the event indicator. The Brier score at a fixed time $t$ is defined as
   \[	BS(t) = \frac{1}{n} \sum^n_{i=1} \left( I(Y_i>t) - \hat{S}(t|X_i) \right)^2	\]
where $\hat{S}(t|X_i)$ is the predicted survival rate conditional on $X_i$ given by the model. 
For right-censored survival data $(X_i, Y_i,\delta_i)$, $i=1,2,...,n$, the definition of the Brier score at time $t$ is  
\[
BS(t) = \frac{1}{n} \sum^n_{i=1} \left[ \hat{S}(t|X_i)^2 I(Y_i \leq t \land \delta_i = 1)\hat{G}_{Y_i}^{-1} + \left(1 - \hat{S}(t|X_i) \right)^2 I(Y_i > t)\hat{G}_{t}^{-1} \right].	
\]
Note that since evaluation of performance on the test set is based on the actual uncensored survival times, $\delta_i = 1$ for all observations. The integrated Brier score is given by
\begin{equation}
IBS = \frac{1}{\max{(Y_i)}} \int^{\max{(Y_i)}}_0 BS(Y)dY 
\end{equation} 
and is usually preferred to the time-dependent Brier score since it gives a summary of the prediction error over the entire study period.  

\subsubsection{Simulation setup}
We use three simulation setups to test the prediction performance of the LTRC trees. Besides the two proposed tree methods, we also include the Cox proportional hazards model in the simulations for comparison. To see how left-truncation matters, we also include the versions that ignore left-truncation of the two tree methods and the Cox model. The three survival families are as follows:

\begin{enumerate}[label=(\roman*)]
	\item Tree structured data as in Section \ref{Sec 3.1}
	\item  $\vartheta = -X_1-X_2$
	\item $\vartheta = -\left[cos\left(\left(X_1+X_2\right)\cdot \pi\right)+\sqrt{X_1+X_2}\right]$
\end{enumerate}
where $\vartheta$ is a location parameter whose value is determined by covariates $X_1$ and $X_2$. In the first setup, data are generated according to the tree structure described in Section $3.2$, so the LTRC trees should perform well. The second and third setups are similar to those in \cite{hothorn2004bagging}. Six independent covariates $X_1,...,X_6$ serve as predictor variables, with $X_2, X_3,X_6$ binary$\{0,1\}$ and $X_1, X_4, X_5$ uniform$[0,1]$. The survival time $T_i$ depends on $\vartheta$ with three different distributions:

\begin{itemize}
	\item Exponential with parameter $\lambda = e^{\vartheta}$ 
	\item Weibull with increasing hazard, scale parameter $\lambda = 10e^{\vartheta}$ and shape parameter $k=2$
	\item Weibull with decreasing hazard, scale parameter $\lambda = 5e^{\vartheta}$ and shape parameter $k=0.5$
\end{itemize}

Note that for the Exponential distribution, the hazard function is

\[	h(t | \mathbf{x}) = \lambda	= e^{\vartheta},\]
while for the Weibull distribution, the hazard function is

 \[	h(t | \mathbf{x}) = kt^{k-1} e^{-k \log \lambda}= kt^{k-1} e^{-k (\vartheta + \log 5)} \hspace{0.2in} \text{or} \hspace{0.2in}  kt^{k-1} e^{-k (\vartheta+\log 10)}.\]
Therefore, in the second setup where $\vartheta = -X_1-X_2$, both satisfy the proportional hazards assumption, with the log hazard linearly dependent on covariates 
 \[	 h(t | \mathbf{x}) = h_0(t)e^{\mathbf{x} \beta}.\]
Thus, the Cox PH model should perform best in this setup.

The third setup is similar to the second except that $\vartheta$ in this setup has a more complex nonlinear structure in terms of covariates ($-\left[cos\left(\left(X_1+X_2\right)\cdot \pi\right)+\sqrt{X_1+X_2}\right]$), which makes the distributions of $T_i$ satisfy neither the Cox PH model nor the tree structure. Such a setup is to test how robust the LTRC trees and Cox PH model are in a real world application where survival time might have a complex structure. 

In all setups, the left truncation time $L$ is generated as independent uniform $[0, U]$, with $U$ taking three different values from $\{1,2,3\}$ to represent different truncation rates. The survival time $T$ is generated as detailed above in each setup. If the generated $T < L$, i.e. survival time is less than left-truncation time, then this observation is discarded. Otherwise, the observation is retained, with censoring time $C = L + D$, where $D$ has an exponential distribution. If $C < T$, then this observation is censored ($\delta = 0$), otherwise the survival time $T$ is observed ($\delta=1$). Note that $D$ and $L$ are both independently generated from $T$ and from each other. Two possible censoring rates, light censoring with about $20\%$ observations being censored and heavy censoring with about $50\%$ observations being censored, are considered in each setting.

The survival time $T_i$ in the test set is also generated according to this process, except that no truncation time or right-censoring time is used, i.e. the survival time $T_i$ is never left-truncated or right-censored. The test set is set to have the same sample size as the training set. We pick the size $N$ from $\{100, 300, 500\}$ to see how the sample size affects performance.

Figures \ref{fig:Boxplot1}-\ref{fig:Boxplot6} give side-by-side IBS boxplots for all three settings with sample size $N=300$. In these figures, there are $3$ rows of mini-plots. The first row represents the left-truncation time $L$ with distribution $U(0, 1)$, the second row represents the truncation time $L$ with distribution $U(0, 2)$ and the third row represents truncation time $L$ with distribution $U(0, 3)$. The methods are numbered as follows:
\begin{enumerate}
	\item LTRCIT
	\item Conditional inference survival tree (HHZ) ignoring left truncation 
	\item LTRCART
	\item Relative risk survival tree \citep{1992relative} ignoring left truncation
	\item Cox proportional hazards model
	\item Cox model ignoring left truncation 
\end{enumerate}
The complete set of IBS boxplots, including cases with sample size $N=100$ and $N=500$ can be found in the supplemental material. 

\begin{sidewaysfigure}
\includegraphics[width=\textwidth]{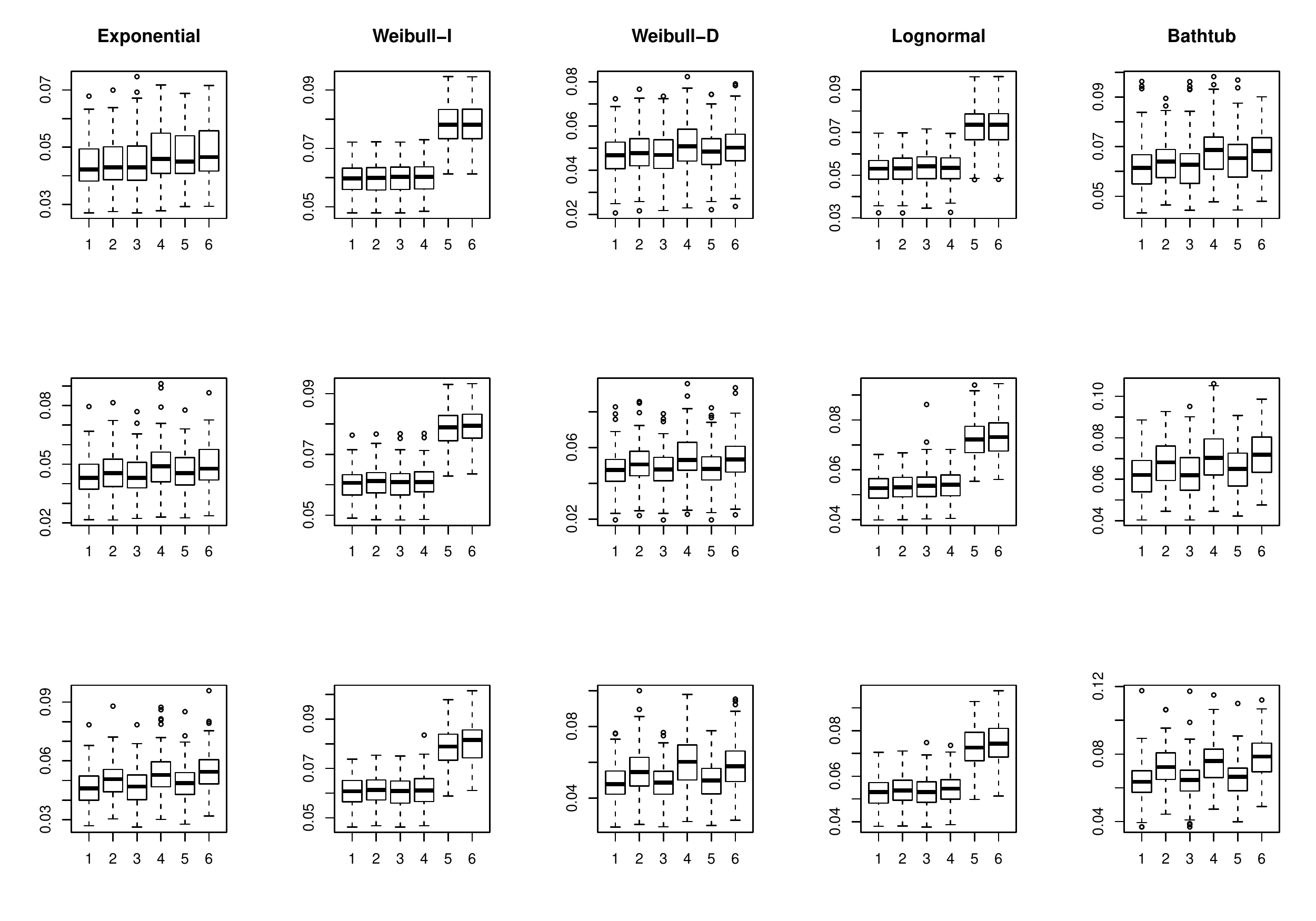}
\caption{\label{fig:Boxplot1}Setting $1$: IBS boxplots with light censoring and $N=300$. Methods are numbered as 
$1$-LTRCIT, $2$-Conditional inference survival tree, $3$-LTRCART, $4$-Relative risk survival tree, $5$-Cox model,
$6$-Cox model ignoring left-truncation. First to third row corresponding to left-truncation time $L\sim U(0,1)$, $L\sim U(0,2)$ and $L\sim U(0,3)$, respectively.}
\end{sidewaysfigure}

\begin{sidewaysfigure}
\includegraphics[width=\textwidth]{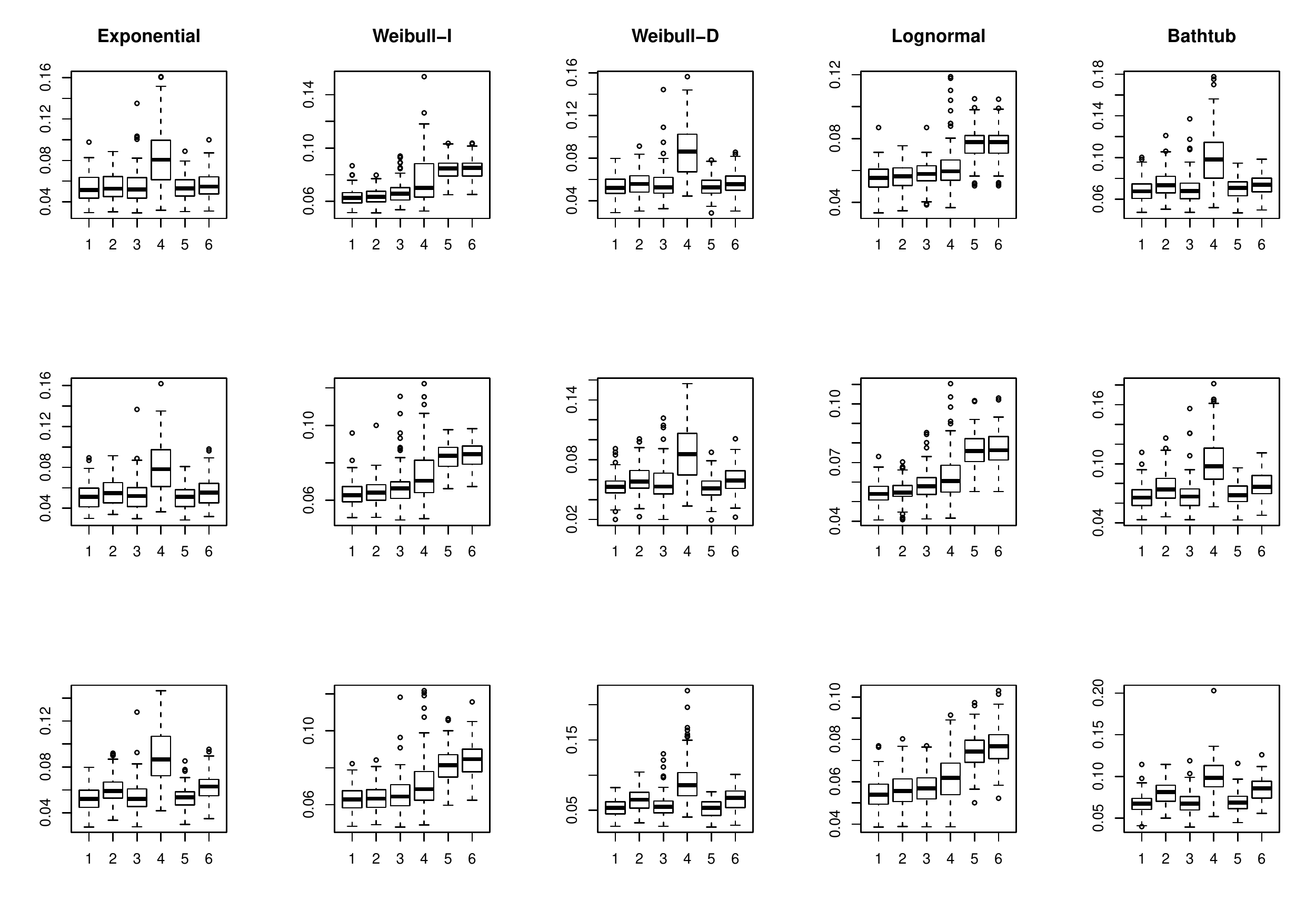}
\caption{\label{fig:Boxplot2}Setting $1$: IBS boxplots with heavy censoring and $N=300$. Methods are numbered as 
$1$-LTRCIT, $2$-Conditional inference survival tree, $3$-LTRCART, $4$-Relative risk survival tree, $5$-Cox model,
$6$-Cox model ignoring left-truncation. First to third row corresponding to left-truncation time $L\sim U(0,1)$, $L\sim U(0,2)$ and $L\sim U(0,3)$, respectively.}
\end{sidewaysfigure}

\begin{sidewaysfigure}
\includegraphics[width=\textwidth]{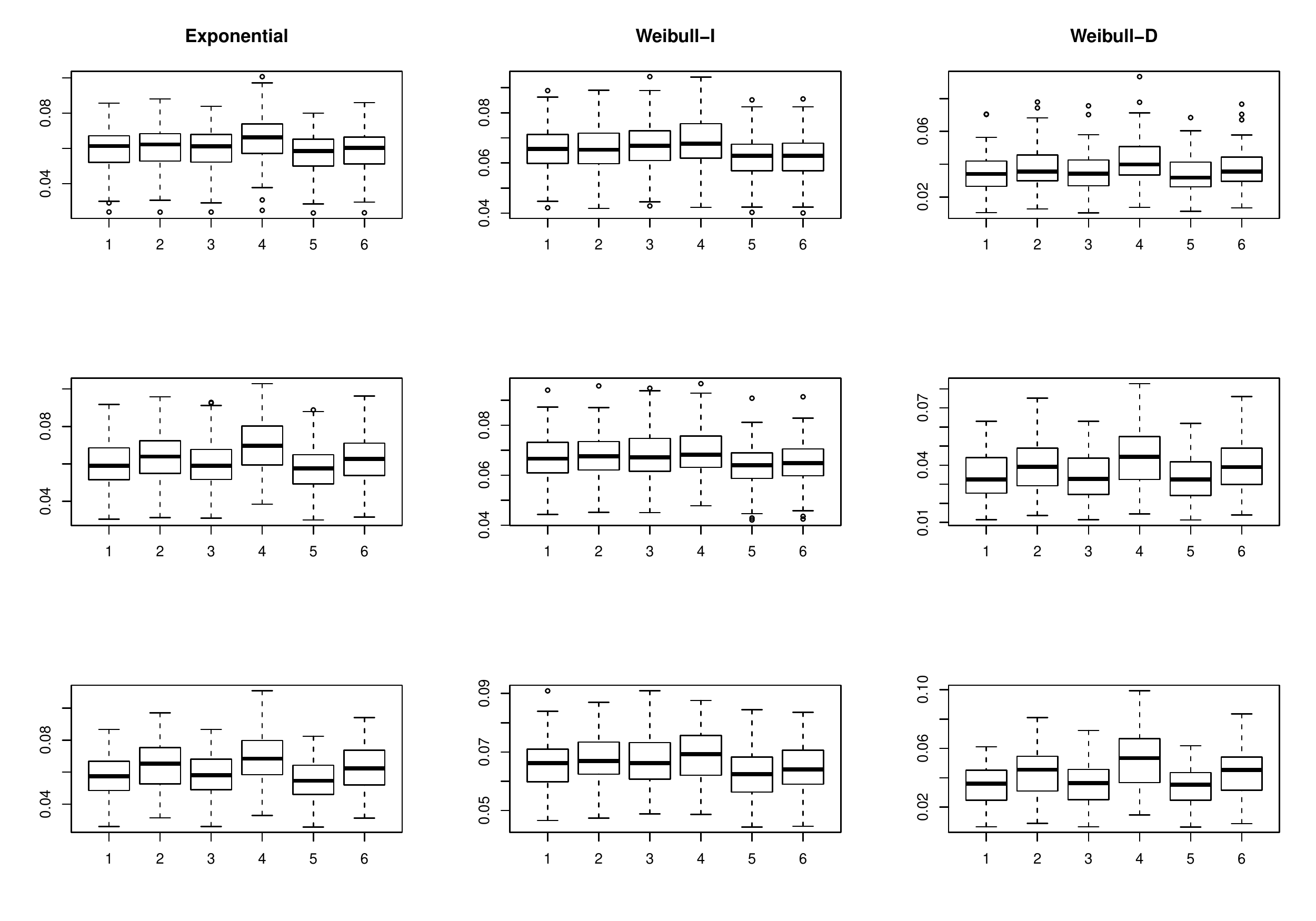}
\caption{\label{fig:Boxplot3}Setting $2$: IBS boxplots with light censoring and $N=300$. Methods are numbered as 
$1$-LTRCIT, $2$-Conditional inference survival tree, $3$-LTRCART, $4$-Relative risk survival tree, $5$-Cox model,
$6$-Cox model ignoring left-truncation. First to third row corresponding to left-truncation time $L\sim U(0,1)$, $L\sim U(0,2)$ and $L\sim U(0,3)$, respectively.}
\end{sidewaysfigure}

\begin{sidewaysfigure}
\includegraphics[width=\textwidth]{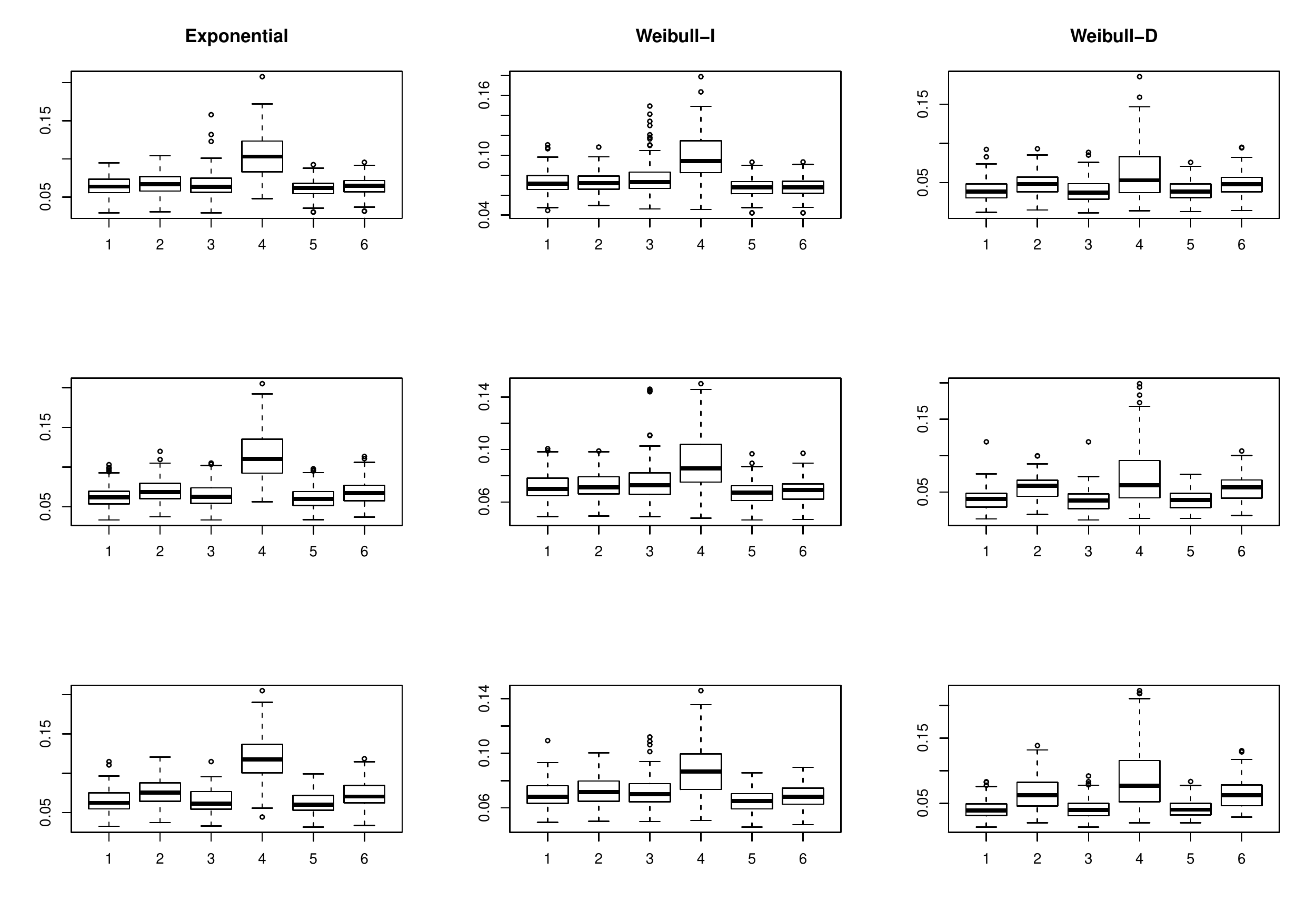}
\caption{\label{fig:Boxplot4}Setting $2$: IBS boxplots with heavy censoring and $N=300$. Methods are numbered as 
$1$-LTRCIT, $2$-Conditional inference survival tree, $3$-LTRCART, $4$-Relative risk survival tree, $5$-Cox model,
$6$-Cox model ignoring left-truncation. First to third row corresponding to left-truncation time $L\sim U(0,1)$, $L\sim U(0,2)$ and $L\sim U(0,3)$, respectively.}
\end{sidewaysfigure}

\begin{sidewaysfigure}
\includegraphics[width=\textwidth]{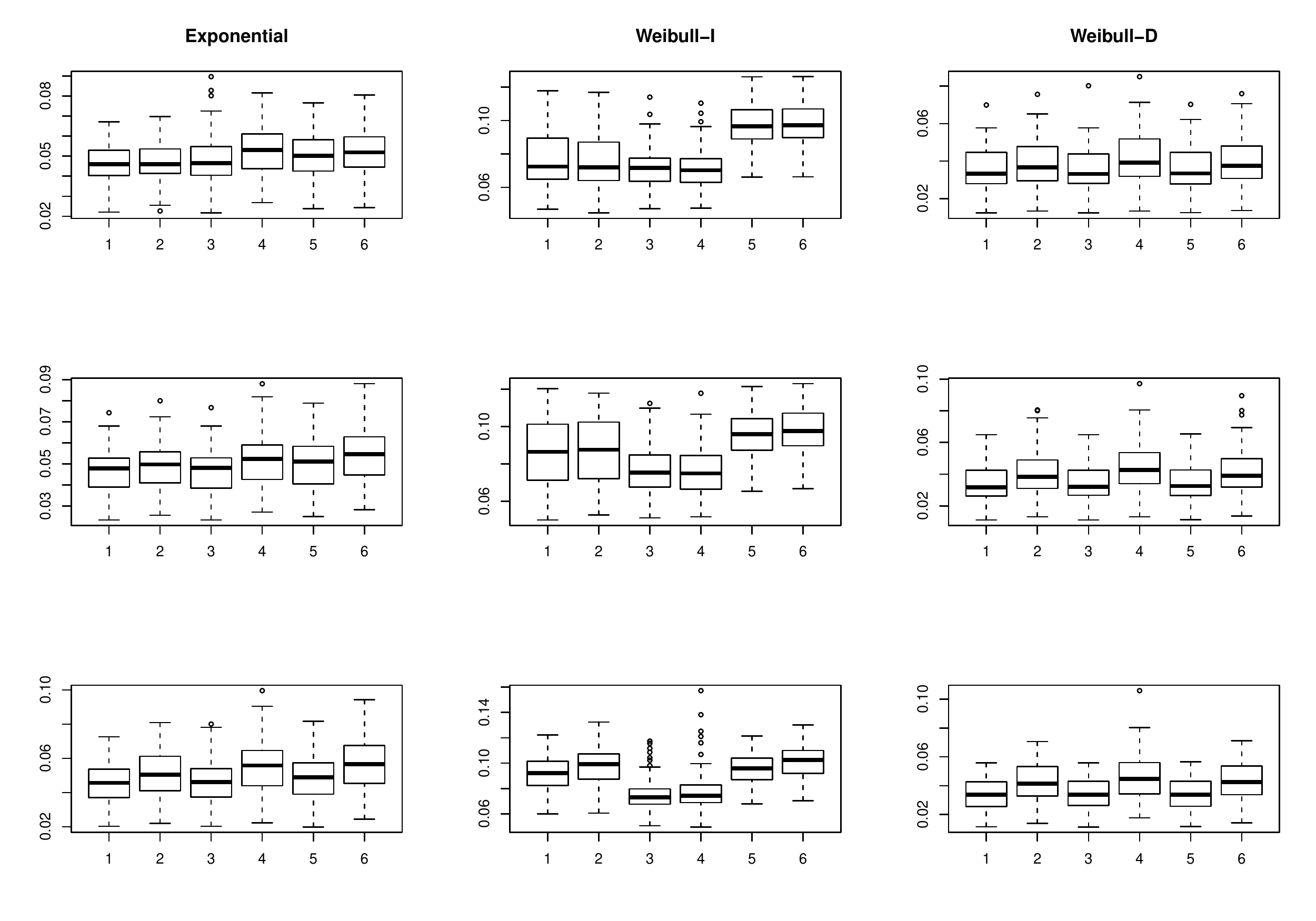}
\caption{\label{fig:Boxplot5}Setting $3$: IBS boxplots with light censoring and $N=300$. Methods are numbered as 
$1$-LTRCIT, $2$-Conditional inference survival tree, $3$-LTRCART, $4$-Relative risk survival tree, $5$-Cox model,
$6$-Cox model ignoring left-truncation. First to third row corresponding to left-truncation time $L\sim U(0,1)$, $L\sim U(0,2)$ and $L\sim U(0,3)$, respectively.}
\end{sidewaysfigure}

\begin{sidewaysfigure}
\includegraphics[width=\textwidth]{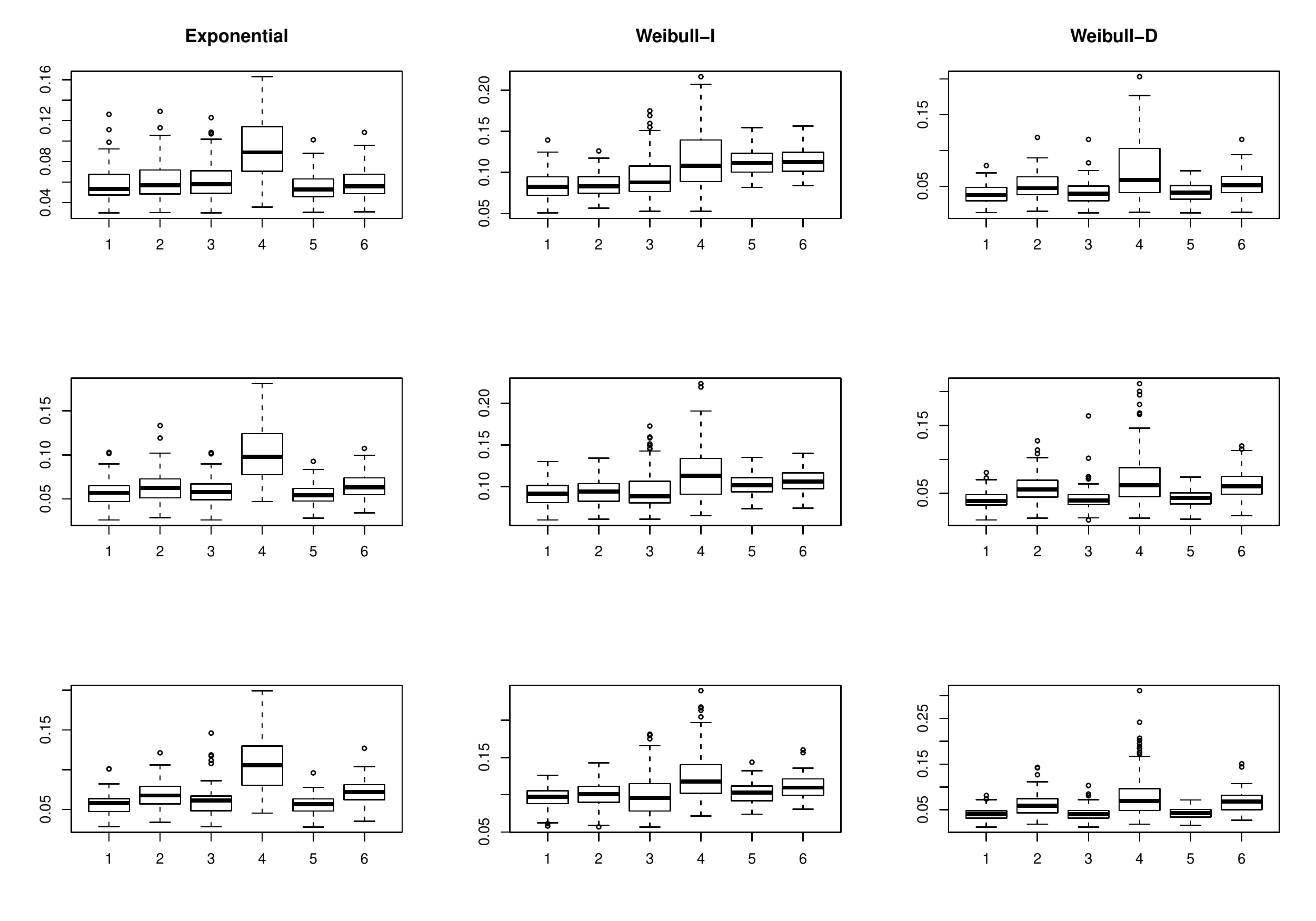}
\caption{\label{fig:Boxplot6}Setting $3$: IBS boxplots with heavy censoring and $N=300$. Methods are numbered as 
$1$-LTRCIT, $2$-Conditional inference survival tree, $3$-LTRCART, $4$-Relative risk survival tree, $5$-Cox model,
$6$-Cox model ignoring left-truncation. First to third row corresponding to left-truncation time $L\sim U(0,1)$, $L\sim U(0,2)$ and $L\sim U(0,3)$, respectively.}
\end{sidewaysfigure}

Since boxplots can look very similar between different methods, we also compare the methods' results using the signed-rank test. The results for each setup are as follows.

\begin{itemize}
	\item \textbf{Setup} $1$ -- \textbf{Tree structured data} \\
	We would expect LTRC trees to perform well in this setup, since the underlying data structure is a tree. This is indeed the case as both LTRC trees perform significantly better than the Cox PH model, regardless of sample size, censoring distribution and left-truncation rate. In terms of sensitivity to censoring rate, LTRCART is most sensitive to heavy censoring, as its performance deteriorates more with higher censoring rate. In contrast, the Cox model is the least sensitive method as its performance is most resistant to heavy censoring. 
	
Ignoring left-truncation results in significantly worse performance in all cases, and its effect becomes more obvious in the high left-truncation rate setting.  
	
LTRCART and LTRCIT have indistinguishable performance in the light censoring case with large sample size. However, when the censoring rate is high or the sample size is small, LTRCIT outperforms LTRCART. That is, LTRCART is more sensitive to heavy censoring and small sample size than is LTRCIT. This conclusion coincides with the results in Section \ref{Sec 3.1}, as one can see that the more frequently a tree can recover the correct tree structure, the better its predictive performance is.  
	
	\item \textbf{Setup} $2$ -- \textbf{Proportional hazards data with log hazard linearly dependent on covariates}\\
	In this setup, the Cox proportional hazards model is expected to perform the best since data are generated to be consistent with the Cox PH model. Indeed, the Cox PH model outperforms both LTRC trees in all settings.
	
	The relative performance of LTRCART tree and LTRCIT tree is a little different in this setup. Here, the two LTRC trees are indistinguishable in the heavy censoring case, while LTRCIT performs better under light censoring. 
	
	Ignoring left-truncation results in worse performance for all methods, and it becomes more obvious with higher left-truncation rate.
	
	\item \textbf{Setup} $3$ -- \textbf{Complex non-linear model}\\
	In this setup, both LTRC trees and the Cox PH model are the wrong model, making this a test of the robustness of each method.

The results show that LTRC trees clearly outperform the Cox PH model for exponential and Weibull increasing hazard distributions, especially in the case of light censoring and large sample size. In fact, both LTRC trees are significantly better than the Cox PH model for all distributions, in the case of light censoring and large sample size. When censoring rate increases and/or sample size decreases, sensitivity to high censoring and small sample size undermine the trees' advantage over the Cox PH model. Nevertheless, the LTRC trees never perform significantly worse than the Cox PH mode, which demonstrates that the LTRC trees have more robust performance than does the Cox PH model.  
\end{itemize} 

In general, increasing sample size favors LTRC trees over the Cox model, especially in the heavy censoring cases. A larger sample size also results in relatively worse performance for methods that ignore left-truncation. The reason is presumably that a larger sample size reduces the variability of all methods, in which case the benefits of accounting for left-truncation stand out. Since trees are relatively unstable compared to the Cox model, they benefit more from a larger sample size than does the Cox model. 

In terms of survival distribution, trees perform relatively better for the Weibull distribution with increasing hazard and the Lognormal distribution than for the Exponential distribution, the Weibull distribution with decreasing hazard and the distribution with bathtub-shaped hazard. Indeed, the Weibull distribution with increasing hazard and the Lognormal distribution look similar to each other (the density peaking at different times for different leaves) in Figure \ref{density}, while the other three distributions share a similar pattern (each leaf's density peaking at roughly the same time). It is apparently easier for trees to separate groups that peak at different times than those that peak at similar times, and therefore trees work better for the Weibull distribution with increasing hazard and Lognormal distribution.          

\section{Real LTRC data application}
The assay of serum free light chain data for $7874$ subjects in the \texttt{R} package \texttt{survival} \citep{survival-package} is used as a data example. It is a random sample containing one-half of the subjects from a study of the relationship between serum free light chain (FLC) and mortality by \cite{dispenzieri2012use}. 
The objective of the study is to determine whether the free light chain (FLC) assay provides prognostic information relevant to the general population. \cite{dispenzieri2012use} concluded that a nonclonal elevation of FLC is a significant predictor of worse overall survival in the general population of persons without plasma cell disorder.

The predictors of interest are
\begin{itemize}
	\item Age
	\item Sex: \hspace{0.05in} F=female, M=male
	\item FLC: \hspace{0.05in} the FLC group for the subject, ranging from $1,2,...,10$ ($1$=lowest decile, $10$=highest decile)
	\item Creatinine: \hspace{0.05in} serum creatinine
\end{itemize}

The original analysis was based on the Cox model including Age as one of the covariates, which showed Age, Sex, FLC top decile and Creatinine were all significant. The response was time from enrollment of study to death/censoring. However, as noted by \cite{klein2003survival}, age is often used as a covariate when it should be used as a left-truncation point. This is particularly true in a mortality study such as this one, since greater age is almost always associated with higher risk of death, making it not very meaningful (or surprising) to have age as a (significant) covariate. Also, the real response of interest should be the subject's life length, not the time from enrollment in the study to death/censoring.   

We analyze this data using LTRC trees with age as left-truncation point and the actual death/censoring time as response. From the top panels of Figures \ref{fig-LTRC-ctree} and \ref{fig-LTRC-rpart}, we can see that both LTRC trees identify the top FLC decile (FLC=$10$) as the most important predictor of overall survival, independent of other factors such as Sex and Creatinine. \cite{dispenzieri2012use} collapsed the data into $2$ groups, $10$th decile vs deciles $1$ through $9$, before analyzing the data using the Cox model, which leads to the same conclusion. Thus, we can see the LTRC tree results are well-aligned with the original result of \cite{dispenzieri2012use}. The two LTRC trees are broadly similar, with LTRCART having more end splits on Creatinine in the left branch compared to the LTRCIT tree, which may be caused by the tendency of LTRCART to split more on continuous variables. 

In contrast, both regular survival tree results (ignoring the left-truncation) seem to underestimate the effect of FLC, with the conditional inference survival tree of HHZ (lower panel of Figure \ref{fig-LTRC-ctree}) identifying FLC as important only for males, and the relative risk survival tree (lower panel of Figure \ref{fig-LTRC-rpart}) completely missing identifying FLC as a significant predictor.

Note that each terminal node of LTRCIT gives the estimated KM curve on that node, while LTRCART shows a single number, the relative risk, on its terminal node, which is the proportion of hazard of that terminal node relative to baseline hazard (root node). A larger number in a terminal node of LTRCART thus means higher hazard rate in that node, which implies a steeper KM curve on the corresponding terminal node of LTRCIT. It is also interesting to note that the KM curves on the terminal nodes of LTRCIT all start dropping at time $50$ rather than $0$. This is because all of the subjects in these data were aged $50$ or greater when entering this study, and therefore the KM curves are left-truncated at time $50$.   

\begin{figure}[H]
\centering
\begin{minipage}{\linewidth} \centering
    \includegraphics[width=0.8\linewidth, height=0.4\textheight]{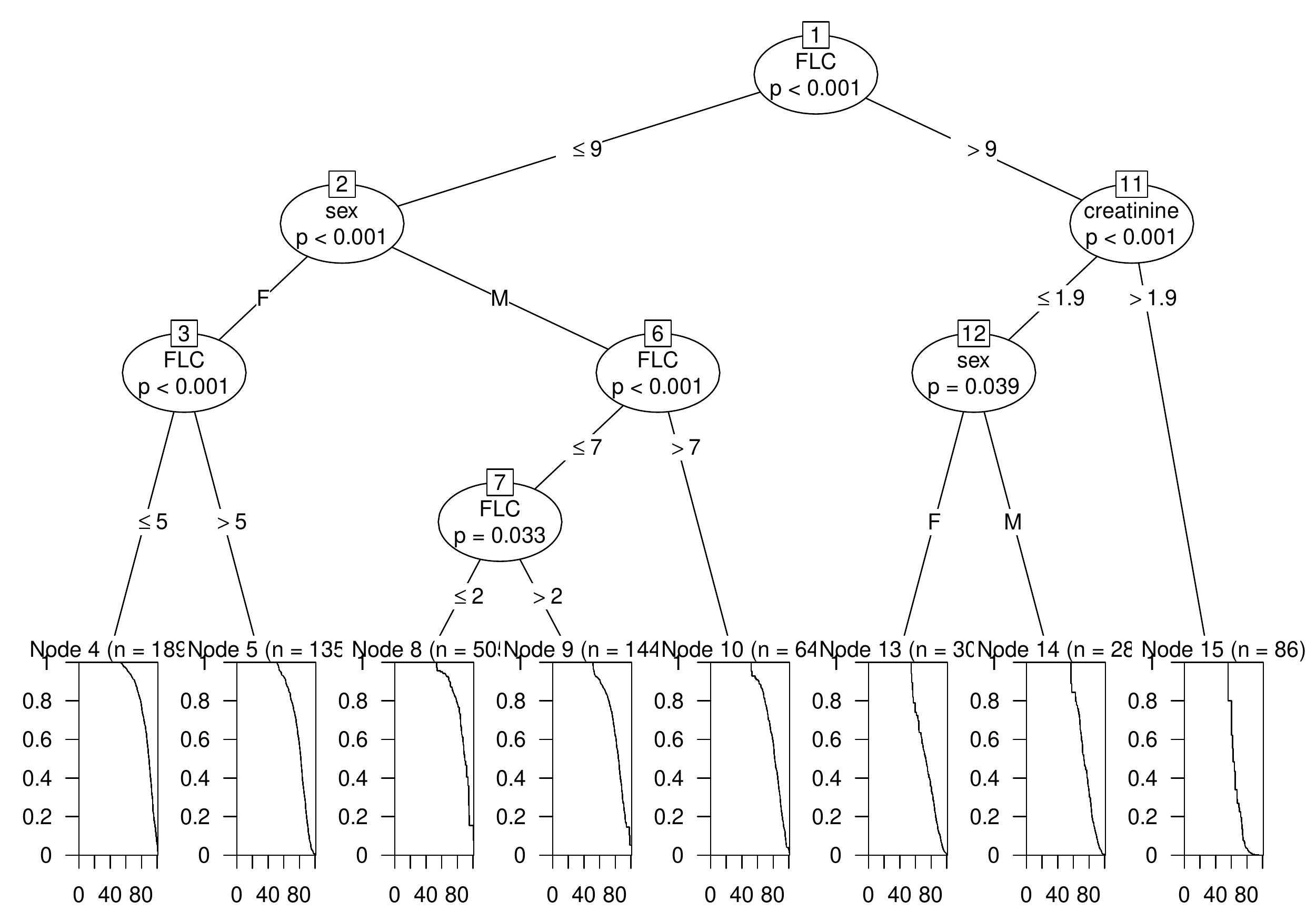}
\end{minipage}
\begin{minipage}{\linewidth} \centering
    \includegraphics[width=0.8\linewidth, height=0.4\textheight]{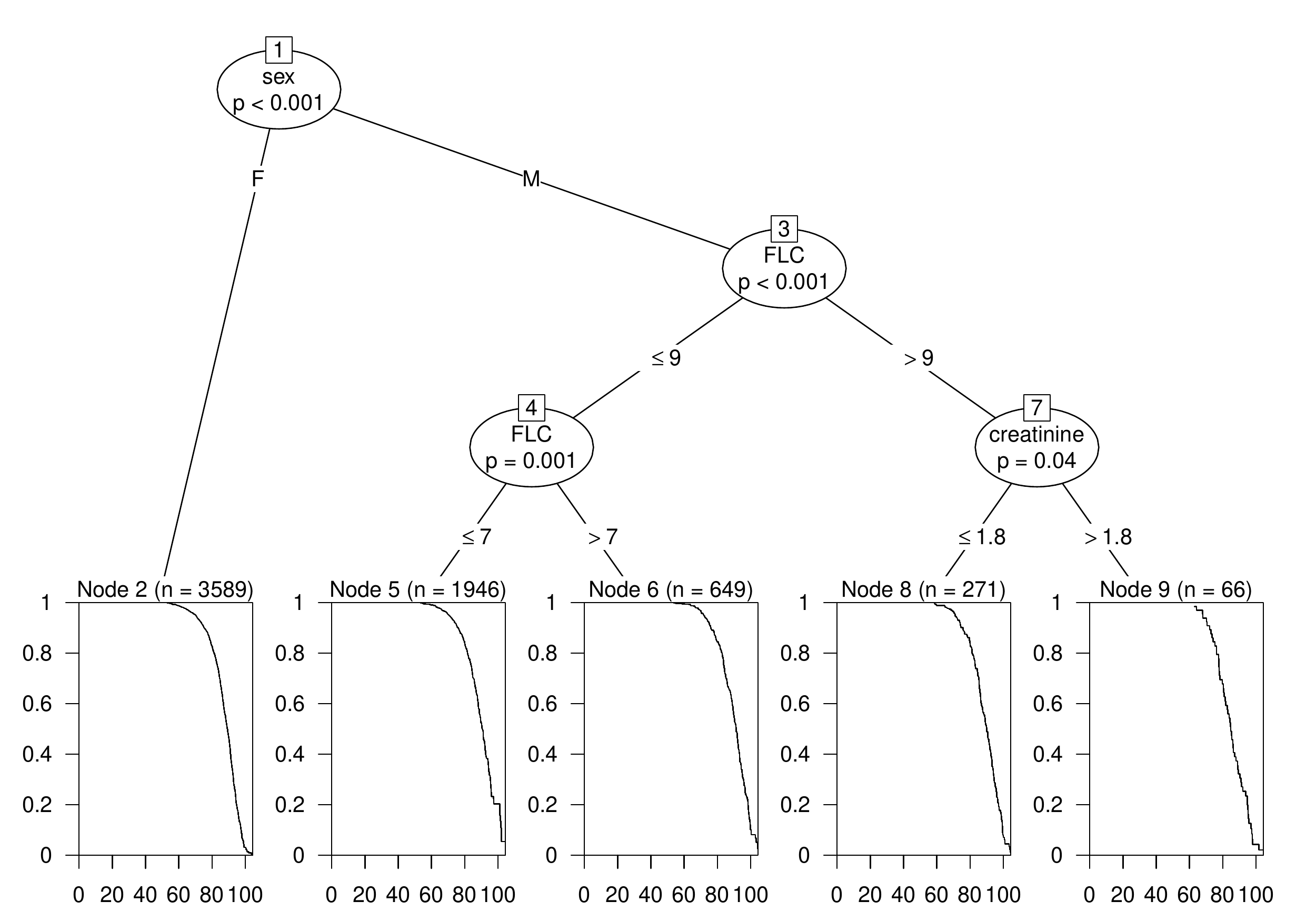}
\end{minipage}
\caption{Upper panel shows the LTRCIT tree for the serum free light chain data; the lower panel shows the conditional inference survival tree ignoring left-truncation.}
\label{fig-LTRC-ctree}
\end{figure}

\begin{figure}[H]
\centering
\begin{minipage}{\linewidth}  \centering
    \includegraphics[width=0.8\linewidth, height=0.4\textheight]{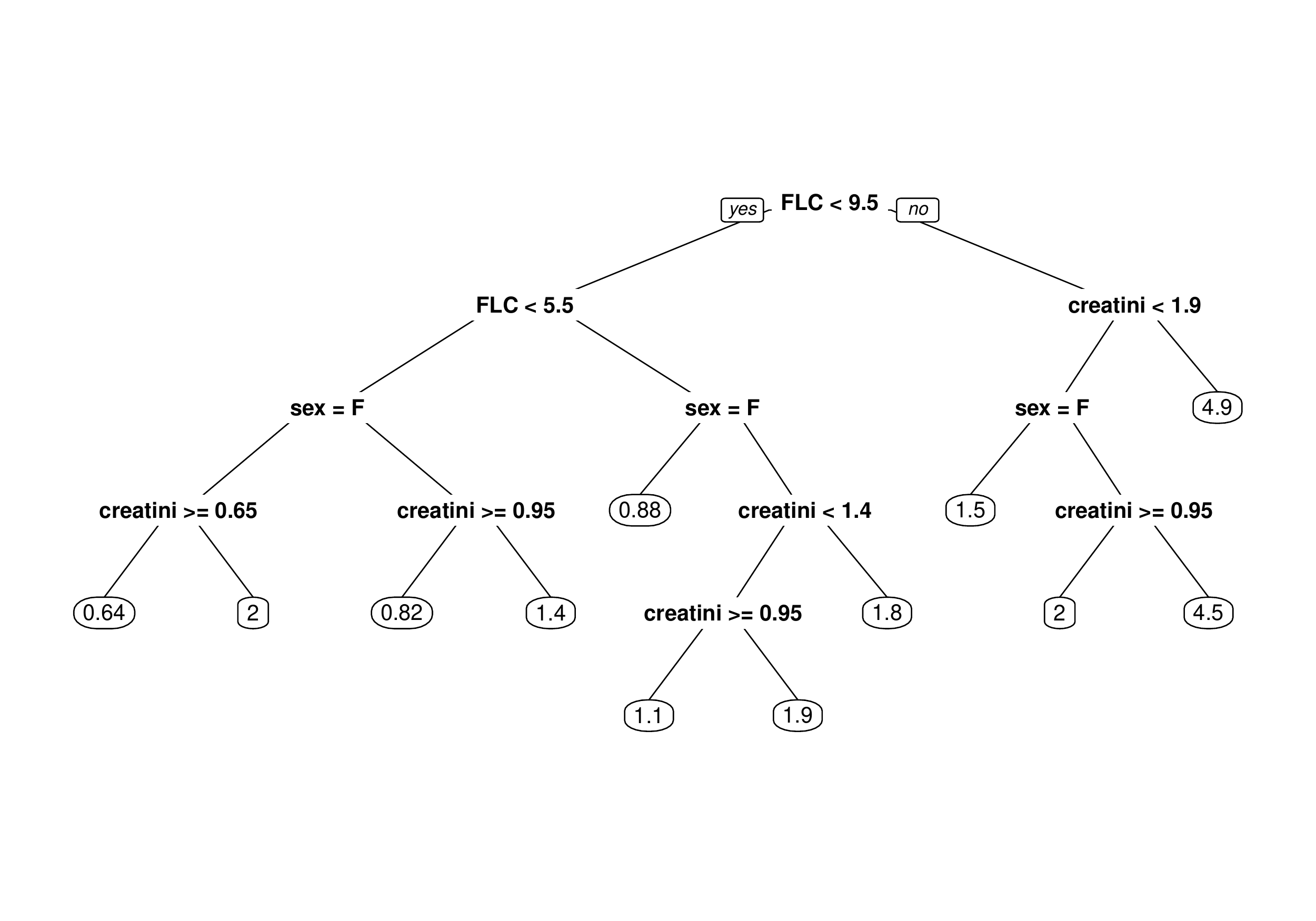}
\end{minipage}
\begin{minipage}{\linewidth}  \centering
    \includegraphics[width=0.9\linewidth, height=0.4\textheight]{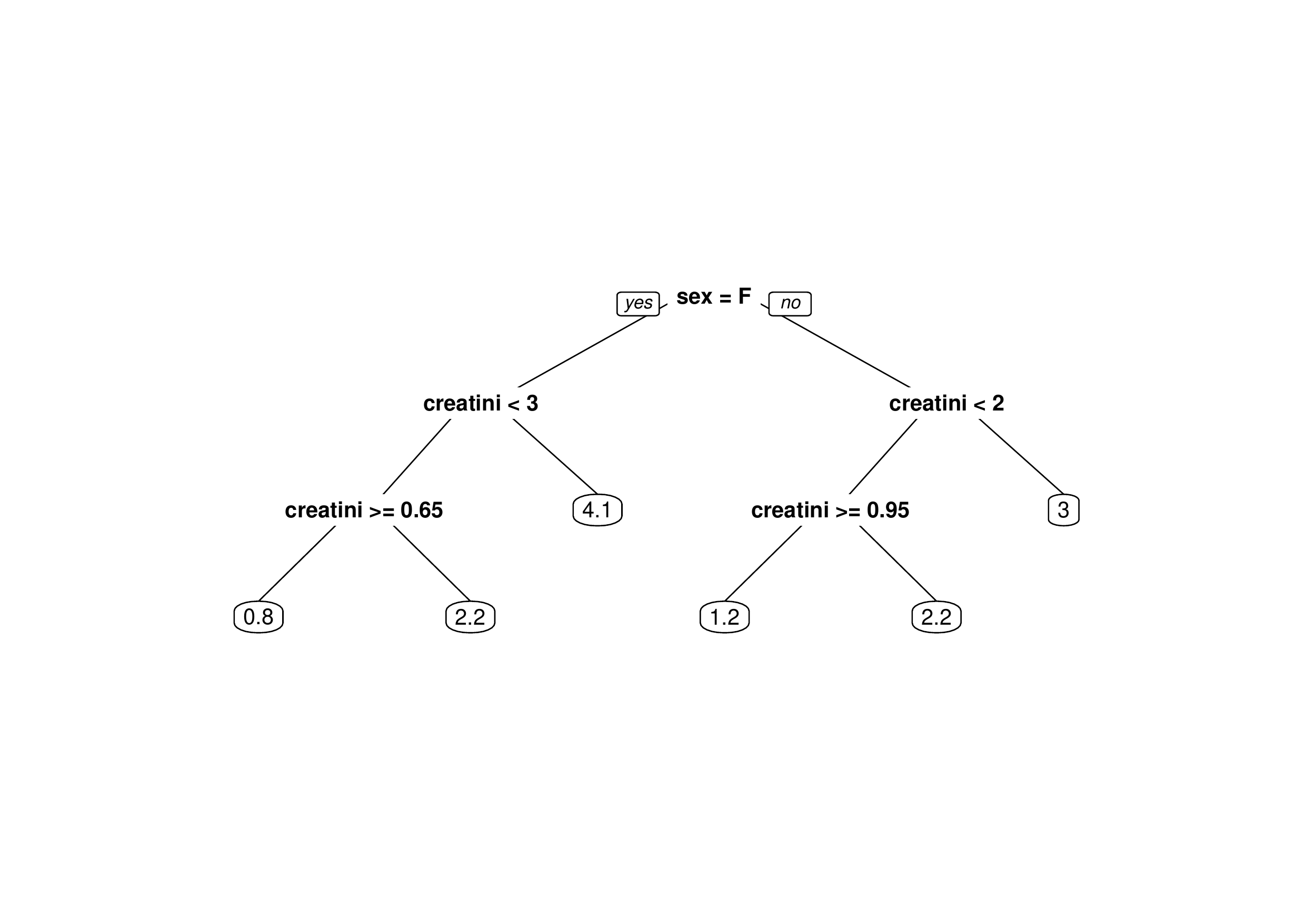}
\end{minipage}
\caption{Upper panel shows the LTRCART relative risk survival tree for the serum free light chain data; the lower panel shows the relative risk survival tree ignoring left-truncation.}
\label{fig-LTRC-rpart}
\end{figure}

\begin{table}[t]
\center
\caption{\label{cox-table}  Cox model on serum free light chain data with Age as delay entry time} 
\begin{tabular}{ lcrcccccrcl}
\hline
 Effect && $coef$ && $exp(coef)$ && $se(coef)$ && $z$ && $p$\\
  \hline
  Sex-Male      && $0.246894$ && $1.280043$  && $0.047972$ && $5.147$ && $2.65e^{-7}$\\    
   FLC       && $0.106503$ && $1.112381$  && $0.008847$ && $12.038$ && $< 2e^{-16}$\\   
Creatinine   && $0.234216$ && $1.263918$  && $0.031452$ && $7.447$ && $9.57e^{-14}$\\   
  \hline
\end{tabular}
\end{table}
If we fit the Cox model with age as left-truncation points, we get the results given in Table \ref{cox-table}.
The Cox model incorporating left-truncation also identifies Sex, FLC and Creatinine as significant risk factors. FLC also seems to be the most significant predictor. However, the Cox model has no way to automatically detect the 10th decile as the most significant factor; in fact, \cite{dispenzieri2012use} used domain expertise to collapse the data into the $2$ groups $10$th decile vs. deciles $1$ through $9$, before using the Cox model for analysis.

 \begin{table}[t]
\center
\caption{\label{cox-strat}  LTRC Cox model on serum free light chain data stratified by Sex} 
\begin{tabular}{ lcrcccccrcl}
\hline
 Effect && $coef$ && $exp(coef)$ && $se(coef)$ && $z$ && $p$\\
  \hline
  Male:FLC    && $0.126748$ && $1.135130$  && $0.009181$ && $13.805$ && $< 2e^{-16}$\\    
Female:FLC    && $0.092362$ && $1.096761$  && $0.009574$ && $9.647$ && $< 2e^{-16}$\\   
Creatinine    && $0.232506$ && $1.261758$  && $0.031466$ && $7.389$ && $1.48e^{-13}$\\   
  \hline
\end{tabular}
\end{table}
To see if FLC is significant for both genders, we fit the stratified Cox model (Table \ref{cox-strat}).
It is clear from Table \ref{cox-strat} that FLC is a significant factor for both genders. This is consistent with the top panels of both Figures \ref{fig-LTRC-ctree} and \ref{fig-LTRC-rpart}, but not the lower panels of Figures \ref{fig-LTRC-ctree} and \ref{fig-LTRC-rpart}, demonstrating how important it is to incorporate left-truncation when it is appropriate. 

\section{Using LTRC trees to fit survival trees with time-varying covariates}
The general strategy to build a time-varying covariates survival tree consists of two steps: first, split each subject into several pseudo-subjects, inside which covariates are time-independent; second, apply the LTRC tree algorithm on those pseudo-subjects to fit a tree. An example illustrates the process.

Assume the survival data (with time-varying covariates) comes in a longitudinal format (the so-called ``long'' format), where each subject may have multiple records of measurements of risk factors during their multiple visits. For example, the top part of Table \ref{patient-eg1} gives the information of a subject that consists of $3$ records. The event--death is observed at time $27$, while the $3$ measurements of Age and CD$4$ are recorded at the beginning and times $10$ and $20$ respectively.
\begin{table}[t]
\center
\caption{\label{patient-eg1} Original and Reformatted data for a patient} 
\begin{threeparttable}
\begin{tabular}{ ccccccccccc}
\hline
 Patient.ID && Age && CD$4$ && Time && Death ($\delta$) \\
  \hline
 $1$ && $45$ && $27$ && $0$ && $0$ &&\\
 $1$ && $45$  && $31$ && $10$ && $0$&& \\
 $1$ && $45$  && $25$ && $20$ && $0$ &&\\  
 $1$ && $-$  && $-$ && $27$ && $1$ &&\\   \hline
 Patient.ID && Age && CD$4$ && Start && End && Death ($\delta$) \\ \hline
 $1$ && $45$ && $27$ && $0$ && $10$ && $0$ \\
 $1$ && $45$  && $31$ && $10$&& $20$ && $0$ \\
 $1$ && $45$  && $25$ && $20$&& $27$ && $1$ \\ \hline
\end{tabular}
 \begin{tablenotes}[para,flushleft]
\footnotesize{The top table gives the original data, while the bottom table shows reformatted data.}
  \end{tablenotes}
\end{threeparttable}
\end{table}
\noindent

The reformatted (transformed) data structure has structure shown at the bottom of Table \ref{patient-eg1}, where each observation (each row) becomes left-truncated (at time Start) and right-censored/event (at time End) data. If we fit the reformatted data in Table \ref{patient-eg1} with an LTRC tree, we effectively get a survival tree that splits on a time-varying covariate. 

The goal is to find intervals such that covariates do not change values inside each interval. If $x(t)$ is changing continuously, infinitely many intervals would be needed to represent data this way. In practice, however, $x(t)$ is typically not monitored all of the time, but rather occasionally, such as when patients are visiting a hospital or clinic. This means that in practice the time-varying covariates are assumed constant between visits.

Such a procedure to process time-varying covariates is not new. In fact, it has been adopted to prepare data in order to fit time-varying covariates in the Cox PH model, and is usually referred to as the Andersen-Gill method \citep{andersen1982cox}. This is more efficient in the tree context than creating the pseudo-subjects each time inside each splitting node as in \cite{bacchetti1995survival}. 

Technically, the definition of a pseudo-subject is as follows:
\begin{definition}
$\forall j \in \{1,2,...,n\}$, the survival information of the $j$th subject $\left(Y_j,\delta_j, x_j(t)|_{0< t \leq Y_j}\right)$ is replaced by a set of $S_j$ pseudo-subjects, with survival information
\[  \left\lbrace \left(L^l_j, R^l_j, \delta^l_j, x_j(t)|_{L^l_j< t \leq  R^l_j}\right)| l=1,...,S_j\right\rbrace\]
where $\left(L^l_j, R^l_j, \delta^l_j, x_j(t)|_{L^l_j< t \leq  R^l_j}\right)$ is the $l$th pseudo-subject and 
\begin{enumerate}
	\item $\bigcup_{l=1}^{S_j} (L^l_j, R^l_j] = (0, Y_j]$
	\item $(L^l_j, R^l_j] \cap (L^h_j, R^h_j] = \emptyset$ if $l \neq h$
	\item $\delta^l_j = \begin{cases} 		
	\delta_j & \text{if}\hspace{0.1in} R^l_j = Y_j \\
	0& \text{if}\hspace{0.1in} R^l_j \neq Y_j
	\end{cases}$
	\item $x(t)$ is constant in each Pseudo-subject, i.e. $\forall l, x_j(t)$ is constant over$(L^l_j, R^l_j]$ 
\end{enumerate}
\end{definition}
\noindent

\subsection{Reasoning about the data reformulation}
Note that one implicit requirement of using the proposed LTRC trees is that the observations in the sample are independent. The reformulation procedure creates several pseudo-subjects from one original observation, and therefore they are not independent, since in our example we know that all of the pseudo-subjects (rows at the bottom of Table \ref{patient-eg1}) have $\delta=0$ except (possibly) for the last one. We now show that they can be treated as independent subjects in terms of contributing to the test statistics inside the algorithms of the proposed LTRC trees.

Without loss of generosity (WLOG), we focus on one specific observation $(Y_i, \delta_i)$. Assume we partition the time interval $[0,Y_i]$ into three segments at time $\tilde{t}_1$ and $\tilde{t}_2$ to create three pseudo-subjects. By definition, the three pseudo-subjects are:
	\begin{enumerate}
		\item $(0, \tilde{t}_1,  \tilde{\delta}_1 =0)$
		\item $(\tilde{t}_1, \tilde{t}_2,  \tilde{\delta}_2 =0)$ 
		\item $(\tilde{t}_2, Y_i,  \delta_i)$
	\end{enumerate}		

We now establish the following lemmas. 

\begin{lemma}  \label{lemma.1} $U_i = \sum^3_{j=1} \widetilde{U}^i_j$, where $U_i$ is the log-rank score of observation $(Y_i, \delta_i)$; $\widetilde{U}^i_j$ is the log-rank score of the $j$th pseudo-subject created from $(Y_i, \delta_i)$, and all pseudo-subjects are treated as if they were independent.
\end{lemma}

\begin{proof}
The log-rank score $U$ for right-censored data  $(Y_i, \delta_i)$ is
\[	U_i = \delta_i + \log \hat{S}(Y_i),\]
where $\hat{S}$ is the KM estimator from the right-censored (original) survival data. 
The log-rank score $\widetilde{U}$ for left-truncated and right-censored data  $(L_i, R_i, \delta_i)$ is $$\widetilde{U}_i = \delta_i + \log \tilde{S}(R_i) - \log \tilde{S}(L_i).$$ The log-rank score $\widetilde{U}$ for the three associated pseudo-subjects are then:
	\begin{enumerate}
		\item $\widetilde{U}^i_1 = \log \tilde{S}(\tilde{t}_1) - \log \tilde{S}(0)$,
		\item $\widetilde{U}^i_2 = \log \tilde{S}(\tilde{t}_2) - \log \tilde{S}(\tilde{t}_1)$, and  
		\item $\widetilde{U}^i_3 = \delta_i + \log \tilde{S}(Y_i) - \log \tilde{S}(\tilde{t}_2)$,
	\end{enumerate}		
where $\tilde{S}$ is the KM estimator from the LTRC survival data (pseudo-subjects). Note that both $\hat{S}$ and $\tilde{S}$ are estimated as if all observations are independent. Note that 
$\tilde{S} = \hat{S}$,
since each can be written as $S(t) = \prod_{t_i \leq t}[1-\frac{d_i}{n_i}] $, where $t_i$ is a distinct event time, $d_i$ is the number of events at $t_i$, and $n_i$ is the risk set at $t_i$. One can easily check that creating  pseudo-subjects does not change $d_i$ or $n_i$ for any $t_i$, and therefore this equality holds.

It is then easy to see that
\begin{equation}
U_i = \widetilde{U}^i_1 +\widetilde{U}^i_2 +\widetilde{U}^i_3 ;
\end{equation}
that is, the sum of log-rank scores of all pseudo-subjects of an original (right-censored) observation is equal to the log-rank score of that observation, by treating the pseudo-subjects as if they were independent.    
\end{proof}

\begin{lemma} \label{lemma.2} $\mathcal{L}_i = \sum^3_{j=1} \widetilde{\mathcal{L}}^i_j$, where $\mathcal{L}_i$ is the contribution of observation $(Y_i, \delta_i)$ to the full log-likelihood and $\widetilde{\mathcal{L}}^i_j$ is the contribution of the $j$th pseudo-subject created from $(Y_i, \delta_i)$ to the full log-likelihood as if the pseudo-subjects were independent.
\end{lemma}

\begin{proof}
From equations $(4)$ and $(5)$, we know that the contribution to the full log-likelihood from an observation $(Y_i,\delta_i)$ is
\[
 \mathcal{L}_i =  \delta_i \log \lambda(Y_i) - \Lambda(Y_i) =  \delta_i \log \lambda(Y_i) - \int^{Y_i}_0 \lambda(\mu)d \mu,
\]
while the contribution to the log-likelihood from an LTRC observation $(L_i, R_i, \delta_i)$ is

\[
\widetilde{\mathcal{L}_i} = \delta_i \log \lambda(R_i) - \left(\Lambda(R_i) - \Lambda\left(L_i\right)\right)  = \delta_i \log \lambda(R_i) - \int^{R_i}_{L_i} \lambda(\mu)d \mu .
\]
Thus, the contributions to the log-likelihood from the three pseudo-subjects are
	\begin{enumerate}
		\item $\widetilde{\mathcal{L}}^i_1 = -[\Lambda(\tilde{t}_1)-\Lambda(0)]$,
		\item $\widetilde{\mathcal{L}}^i_2 = -[\Lambda(\tilde{t}_2)-\Lambda(\tilde{t}_1)]$, and 
		\item $\widetilde{\mathcal{L}}^i_3 = \delta_i \log \lambda(Y_i)-[\Lambda(Y_i)-\Lambda(\tilde{t}_2)]$,
	\end{enumerate}	

It can be shown that the hazard function $\lambda$ and cumulative hazard function $\Lambda$ are the same in these two equations using the same argument as in the previous proof, i.e.  distinct event times $(t_i)$, the number of deaths at $t_i$ $(d_i)$, and the risk set at $t_i$ $(n_i)$ are all unaltered by creating pseudo-subjects. Therefore, 
\[	 \mathcal{L}_i = \widetilde{\mathcal{L}}^i_1 + \widetilde{\mathcal{L}}^i_2 + \widetilde{\mathcal{L}}^i_3	\]
\end{proof}

\begin{theorem} \label{theory.1}
For an observation $(Y_i, \delta_i)$, one can use three LTRC observations, which are same as the three pseudo-subjects created from $(Y_i, \delta_i)$, as a substitute for the contribution of observation $(Y_i, \delta_i)$ to the test statistics in the proposed tree algorithms, treating these observations as if they are independent.
\end{theorem}

\begin{proof}
We will first prove the result for the LTRCIT tree. WLOG, we assume that the observation $(Y_i, \delta_i)$ is in group A, a subnode of a tree. Then the test statistic is $$ T_A = \sum^m_{i=1}U_i $$ with the contribution of observation $(Y_i, \delta_i)$ being log-rank score $U_i$. Let $\widetilde{U}^i_1$, $\widetilde{U}^i_2$ and $\widetilde{U}^i_3$ denote the log-rank scores of the three independent LTRC observations, respectively. Because they are the same as the pseudo-subjects and are independent, by Lemma \ref{lemma.1}, $U_i =\sum^3_{j=1}\widetilde{U}^i_j$. This means that replacing the observation $(Y_i, \delta_i)$ with the three independent LTRC observations does not change the test statistics $T_A$.

To prove the case for LTRCART, one can directly use Lemma \ref{lemma.2}. Since everything in this algorithm begins with the full likelihood function, one only needs to show that replacing $(Y_i, \delta_i)$ with the three independent LTRC observations does not change the full likelihood, which is obvious from Lemma \ref{lemma.2}. 
\end{proof}
 
Since the observation $(Y_i, \delta_i)$ is replaced by the three pseudo-subjects inside the LTRC tree algorithms, Theorem \ref{theory.1} implies that a time-varying covariates tree can be constructed by treating these pseudo-subjects as if they were independent. 

\section{Properties of the time-varying covariates survival trees}

In this section, the proposed LTRC trees are used to fit survival trees with time-varying covariates. Unlike in the time-independent covariates case, simulating survival time with time-varying covariates is non-trivial. Authors such as \cite{leemis1990variate}, \cite{zhou2001understanding}, \cite{sylvestre2008comparison}, \cite{austin2012generating} and \cite{hendry2014data} have proposed methods to simulate survival time with time-varying covariates under the Cox proportional hazard model. We will follow the method proposed by \cite{austin2012generating} in our simulation study in this section, primarily because of its convenient closed-form expression for simulating survival time.

In \cite{austin2012generating}, a single time-varying covariate $z(t)$ and other time-independent covariates $x$ are included in the proportional hazards model. Letting $\beta$ denote the vector of coefficients associated with $x$ and letting $\beta_z$ be the coefficient of $z(t)$, the hazard function is then

\[	h(t,x,z(t)) = h_0(t)e^{\beta x+\beta_z z(t)}.\] 

Three types of time-varying covariate are considered in \cite{austin2012generating}: a dichotomous time-varying covariate that can change value from untreated to treated at most once (e.g, organ transplant); a continuous time-varying covariate such as cumulative exposure to a fixed dose of radiation; and a dichotomous time-varying covariate that can move from untreated to treated and back to untreated (e.g, drug use status). Since in practice the covariates are measured at time intervals, it can be seen as changing values in time as a step function. This simplifies the model but is generally enough for practical purposes \citep{zhou2001understanding}.

Closed-form formulas are derived to simulate survival times from three commonly used distributions: Exponential, Weibull and Gompertz distribution, respectively, since these distributions share the proportional hazards assumption. The baseline hazard function in terms of parameters for each distribution can be described as follows:
\begin{itemize}
	\item $h_0(t) = \lambda$ for Exponential distribution with parameter $\lambda$
	 
	\item $h_0(t) = \lambda \nu t^{\nu-1}$ for Weibull distribution with scale parameter $\lambda$ and shape parameter $\nu$ 
	
	\item $h_0(t) = \lambda \exp(\alpha t)$ for Gompertz distribution with scale parameter $\lambda$ and shape parameter $\alpha$
\end{itemize}

The survival times are then generated as follows. 
\begin{enumerate}[label=(\Roman*)]

\item \textbf{Time-varying covariate with single change -- dichotomous type I}\\
Let $t_0$ denote the time point at which the time-varying covariate $z(t)$ changes from untreated ($Z=0$) to treated ($Z=1$). That is, $z(t) = 0$ for $t < t_0$ and $z(t) = 1$ for $t \geq t_0$. The survival time can be generated for each distribution using its corresponding inverse cumulative function. \cite{austin2012generating} gave closed-form formulas for generating survival times of the three distributions. Details can be found in the supplemental material.

 \item \textbf{Time-varying covariate with multiple changes -- dichotomous type II}\\
Assume all subjects are untreated at $t=0$. Let $t_1$ denote the first time at which $z(t)$ changes from $Z=0$ to $Z=1$; let $t_2$ denote the time at which $z(t)$ changes from $Z=1$ back to $Z=0$; and let $t_3$ denote the time at which $z(t)$ changes from $Z=0$ to $Z=1$ again. There are thus three possible switches between treatment status for each subject. The survival time for this setting is simulated using the closed-form formulas that can be found in the supplemental material.

\end{enumerate}

\subsection{Recovering the correct tree structure}
Similarly to Section \ref{Sec 3.1}, in this section we use simulations to evaluate the proposed trees' ability to recover the correct tree structure, with the covariates in this section possibly being time-varying. 

Five independent covariates $X_1,...,X_5$ are included in the regression, with $X_1,X_3$ being time-independent covariates and $X_2$, $X_4$, $X_5$ being time-varying covariates. $X_1$, $X_2$, $X_3$ are binary$\{0,1\}$, while $X_4$ is $U[0,1]$ and $X_5$ is ordinal takes value randomly from set $\{1, 2, 3, 4, 5\}$. The true model is 

\begin{equation}
h(t,x,z(t)) = h_0(t)e^{\beta x + \beta_z z(t)} = h_0(t)e^{\beta I_{\{X_1 = 1\}} + \beta_z I_{\{X_2 = 1\}}}
\end{equation}
where $h_0(t)$ depends on the specific distribution. The parameters in each distribution are given in Table \ref{TV-parameter}, where the Exponential distribution has constant hazard rate, the Weibull distribution has decreasing hazard rate with time and the Gompertz distribution has increasing hazard rate with time.  

\begin{table}[t]
\center
\caption{\label{TV-parameter} Parameters and Coefficients for each distribution } 
\begin{tabular}{lcccccccc}
\hline
 		 && $\beta$ && $\beta_z$ && Scale $\lambda$ && Shape $\alpha$/$\nu$ \\ \hline
Exponential && $0.8$ &&   $1.4$  &&   $0.1$    &&   $-$   \\
    Weibull && $0.9$ &&   $1.6$  &&   $0.3$    &&   $0.8$    \\
   Gompertz && $1.2$ &&   $2.0$  &&   $0.2$    &&   $0.1$  \\ \hline
\end{tabular}
\end{table}

Note that only $X_1$ and $X_2$ determine the actual survival distributions, so the true underlying data structure can be presented by a tree in two different ways, as shown in Figure \ref{fig-TV-struct}.

\begin{figure}[t]
  \begin{center}
   \captionsetup{justification=centering}
      \includegraphics[width=6 in,height =3 in]{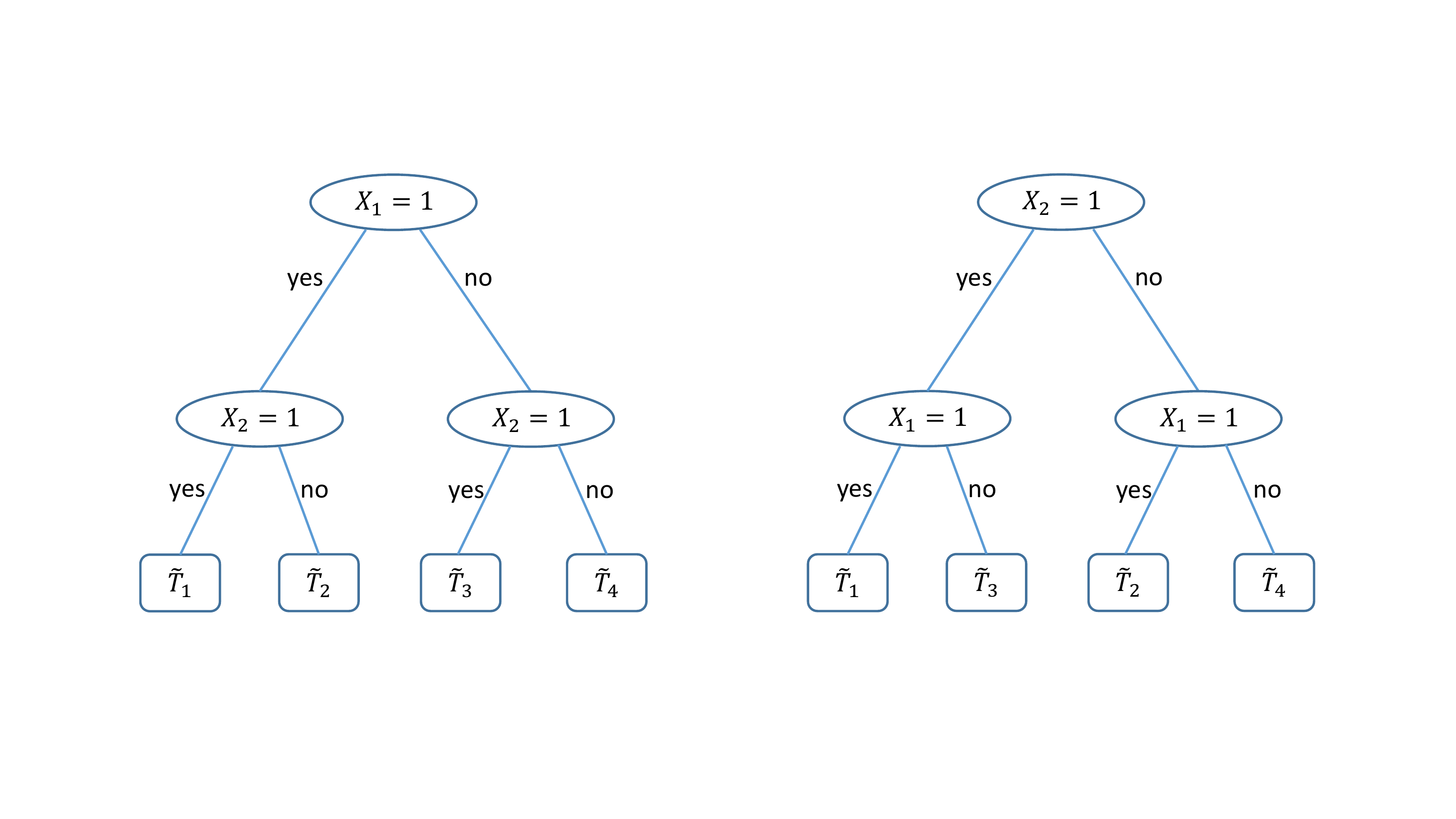}
   \end{center}
     \caption{\label{fig-TV-struct} Underlying data structure with binary split variables} 
 \end{figure} 

The binary variable $X_2$ is the time-varying covariate $z(t)$ in \cite{austin2012generating}, which changes value at $t_0$ for dichotomous type I and changes value at $t_1$, $t_2$, $t_3$ for dichotomous type II. In this way the value of $X_2$ is completely determined by time $t_0$ or time $(t_1, t_2, t_3)$. For each subject, the values of covariates and survival time are generated as follows:

\begin{enumerate}
	\item Randomly generate $X_1$ from $\{0,1\}$ and $\mathit{u} \sim U(0,1)$.
	\item Randomly generate $t_0$ or $(t_1, t_2, t_3)$ from $U(0.6,6)$.
	\item Calculate survival time $T$ from the closed-form formulas based on the values of $X_1$, $\mathit{u}$ and $t_0$ or $(t_1, t_2, t_3)$. 
	\item Split the subject at time $t_0$ or $\{t_1, t_2, t_3\}$ to create pseudo-subjects.
	\item Randomly pick the values of $X_3$, $X_4$, $X_5$ on each pseudo-subject.
	\item Generate censoring time $C \sim \exp (\lambda_D)$. If $C < T$, then this subject is censored with final observation time $C$; otherwise, it is uncensored with observed event time $T$. 
\end{enumerate}

Three levels of censoring rate, $0\%$, $20\%$ and $50\%$ are considered for each distribution. The censoring parameter $\lambda_D$ is chosen for each distribution to make sure these censoring rates hold. To test the effect of sample size on the performance of the proposed trees, three different numbers of subjects, $N = 100, 300, 500$ are tested in the simulations. For each setting, we run 1,000 simulation trials and report the percentage of the time the correct tree structure is recovered by the proposed LTRC trees, as well as the percentage of the time $X_1$ and $X_2$ are identified as the risk factors (splitting variables). Tables \ref{tab-typeI} and \ref{tab-typeII} show the results for dichotomous type I and dichotomous type II, respectively.

\begin{table}[t]
\center
\caption{\label{tab-typeI} Percentage of times correct tree structure recovered--Type I} 
\hspace*{-2cm}
\begin{threeparttable}
\scalebox{0.75}{%
\begin{tabular}{ccccccccccccccccccccccccc}
\hline
Censoring rate && \multicolumn{7}{c}{Exponential} && \multicolumn{7}{c}{Weibull} &&  \multicolumn{7}{c}{Gompertz} \\  \cline{3-9} \cline{11-17} \cline{19-25} 
$\%$ && \multicolumn{3}{c}{\footnotesize{LTRCIT}}&&\multicolumn{3}{c}{\footnotesize{LTRCART}} && \multicolumn{3}{c}{\footnotesize{LTRCIT}}&&\multicolumn{3}{c}{\footnotesize{LTRCART}} && \multicolumn{3}{c}{\footnotesize{LTRCIT}}&&\multicolumn{3}{c}{\footnotesize{LTRCART}} \\  \cline{3-5}  \cline{7-9}  \cline{11-13}  \cline{15-17}  \cline{19-21} \cline{23-25}
$N=100$ &&  $X_1$&$X_2$&S  &&  $X_1$&$X_2$&S && $X_1$&$X_2$&S  && $X_1$&$X_2$&S   &&  $X_1$&$X_2$&S && $X_1$&$X_2$&S \\ \hline

$0$  &&  $73.5$&$80.7$&$4.8$  &&  $69.4$&$81.9$&$9.4$ && $85.4$&$75.5$&$1.7$ && $81.6$&$92.1$&$13.9$ && $98.2$&$98.2$&$3.7$  &&  $97.2$&$100.0$&$8.2$ \\
$20$ &&  $59.6$&$74.4$&$1.9$  &&  $60.2$&$72.7$&$4.4$ && $77.8$&$62.2$&$0.9$ && $74.0$&$83.0$&$5.4$ && $95.7$&$92.3$&$0.9$  &&  $93.2$&$99.0$&$3.6$ \\
$50$ &&  $41.5$&$60.6$&$0.7$  &&  $38.4$&$51.4$&$1.2$ && $58.5$&$25.6$&$0.1$ && $50.3$&$39.0$&$0.0$ && $86.7$&$55.3$&$0.0$  &&  $81.4$&$72.4$&$0.2$ \\ \hline
$N=300$ &&  $$&$$&$$  &&  $$&$$&$$  &&  $$&$$&$$  && $$&$$&$$   &&  $$&$$&$$  &&  $$&$$&$$ \\ \hline
$0$  &&  $100$&$100$&$84.4$  &&  $99.5$&$100$&$66.7$ && $100$&$100$&$68.8$ && $99.9$&$100$&$70.1$ && $100$&$100$&$71.2$  &&  $100$&$100$&$76.0$ \\
$20$ && $99.7$&$100$&$75.8$  &&  $97.1$&$99.9$&$53.3$&& $100$&$100$&$56.4$ && $99.2$&$100$&$57.4$ && $100$&$100$&$58.2$  &&  $100$&$100$&$61.4$ \\
$50$ &&$95.6$&$99.6$&$44.9$  &&  $84.6$&$96.9$&$26.3$&& $99.4$&$90.1$&$19.3$ && $97.3$&$92.9$&$17.5$ && $100$&$100$&$45.4$  &&  $99.9$&$100$&$30.1$ \\ \hline
$N=500$ &&  $$&$$&$$  &&  $$&$$&$$  &&  $$&$$&$$  && $$&$$&$$   &&  $$&$$&$$  &&  $$&$$&$$ \\ \hline
$0$  &&  $100$&$100$&$89.0$  &&  $100$&$100$&$85.7$ && $100$&$100$&$86.9$ && $100$&$100$&$87.6$ && $100$&$100$&$89.1$  &&  $100$&$100$&$92.2$ \\
$20$ &&  $100$&$100$&$89.1$  &&  $100$&$100$&$83.1$ && $100$&$100$&$82.6$ && $100$&$100$&$79.5$ && $100$&$100$&$85.2$  &&  $100$&$100$&$89.2$ \\
$50$ &&  $99.9$&$100$&$81.8$ &&  $98.0$&$99.9$&$62.0$ && $100$&$99.9$&$61.3$ && $99.9$&$99.5$&$47.8$ && $100$&$100$&$75.7$  &&  $100$&$100$&$61.6$ \\ \hline 
\end{tabular} 
\hspace*{-2cm}
} 
 \begin{tablenotes}[para,flushleft]
\footnotesize {``S'' means correct tree structure as in Figure \ref{fig-TV-struct}}.
  \end{tablenotes}
\end{threeparttable}
\end{table}
\noindent

\begin{table}[t]
\center
\caption{\label{tab-typeII} Percentage of times correct tree structure recovered--Type II} 
\begin{threeparttable}
\scalebox{0.75}{%
\hspace*{-2cm}
\begin{tabular}{ccccccccccccccccccccccccc}
\hline
Censoring rate && \multicolumn{7}{c}{Exponential} && \multicolumn{7}{c}{Weibull} &&  \multicolumn{7}{c}{Gompertz} \\  \cline{3-9} \cline{11-17} \cline{19-25} 
$\%$ && \multicolumn{3}{c}{\footnotesize{LTRCIT}}&&\multicolumn{3}{c}{\footnotesize{LTRCART}} && \multicolumn{3}{c}{\footnotesize{LTRCIT}}&&\multicolumn{3}{c}{\footnotesize{LTRCART}} && \multicolumn{3}{c}{\footnotesize{LTRCIT}}&&\multicolumn{3}{c}{\footnotesize{LTRCART}} \\  \cline{3-5}  \cline{7-9}  \cline{11-13}  \cline{15-17}  \cline{19-21} \cline{23-25}
$N=100$ &&  $X_1$&$X_2$&S  &&  $X_1$&$X_2$&S && $X_1$&$X_2$&S  && $X_1$&$X_2$&S   &&  $X_1$&$X_2$&S && $X_1$&$X_2$&S \\ \hline

$0$  && $74.3$&$94.3$&$9.8$  &&  $72.8$&$94.7$&$11.2$ &&  $83.8$&$94.6$&$14.3$  && $79.9$&$97.6$&$18.9$&& $97.6$&$99.9$&$35.7$  &&  $94.6$&$100$&$33.7$ \\
$20$ && $62.2$&$92.8$&$4.6$  &&  $61.9$&$90.9$&$5.1$  &&  $75.6$&$85.0$&$4.6$  && $71.5$&$92.5$&$8.4$  && $94.6$&$98.9$&$23.3$  &&  $91.3$&$100$&$22.7$ \\
$50$ && $44.0$&$75.6$&$1.0$  &&  $43.9$&$71.9$&$2.1$  &&  $56.7$&$42.4$&$0.3$  && $51.2$&$57.1$&$1.3$  && $75.2$&$77.5$&$1.8$  &&  $72.2$&$87.4$&$3.6$ \\ \hline
$N=300$ &&  $$&$$&$$  &&  $$&$$&$$  &&  $$&$$&$$  && $$&$$&$$   &&  $$&$$&$$  &&  $$&$$&$$ \\ \hline
$0$  && $100$&$100$&$79.4$  &&  $99.5$&$100$&$61.5$  &&  $100$&$100$&$86.3$  && $99.8$&$100$&$82.8$   && $100$&$100$&$89.3$  &&  $100$&$100$&$90.6$ \\
$20$ && $100$&$100$&$76.5$  &&  $98.3$&$100$&$54.4$  &&  $100$&$100$&$74.8$  && $99.8$&$100$&$63.4$   && $100$&$100$&$88.8$  &&  $100$&$100$&$87.3$ \\
$50$ && $96.0$&$100$&$55.2$ &&  $84.6$&$99.9$&$30.0$ &&  $99.4$&$98.6$&$43.3$  && $96.2$&$98.9$&$29.7$&& $100$&$100$&$66.5$  &&  $99.9$&$100$&$49.4$ \\  \hline
$N=500$ &&  $$&$$&$$  &&  $$&$$&$$  &&  $$&$$&$$  && $$&$$&$$   &&  $$&$$&$$  &&  $$&$$&$$ \\ \hline
$0$  && $100$&$100$&$87.7$  &&  $100$&$100$&$84.4$  &&  $100$&$100$&$89.3$  && $100$&$100$&$91.7$ && $100$&$100$&$90.2$  &&  $100$&$100$&$93.2$ \\
$20$ && $100$&$100$&$89.7$  &&  $99.9$&$100$&$82.2$ &&  $100$&$100$&$89.9$  && $100$&$100$&$88.7$ && $100$&$100$&$88.9$  &&  $100$&$100$&$91.4$ \\ 
$50$ && $99.9$&$100$&$82.7$ &&  $98.9$&$100$&$64.0$ &&  $100$&$100$&$75.5$  && $99.9$&$100$&$56.7$&& $100$&$100$&$87.4$  &&  $100$&$100$&$82.2$ \\ \hline 
\end{tabular} 
\hspace*{-2cm}
} 
 \begin{tablenotes}[para,flushleft]
\footnotesize {``S'' means correct tree structure as in Figure \ref{fig-TV-struct}}.
  \end{tablenotes}
\end{threeparttable}

\end{table}
\noindent

It is clear from Tables \ref{tab-typeI} and \ref{tab-typeII} that both the sample size and the censoring rate have a strong impact on the trees' ability to recover the correct tree structure and identify the risk factors. A small sample size makes it difficult for trees to recover the tree structure and identify the risk factors. While a high censoring rate has a quite consistent effect of deteriorating the trees' ability to recover the correct tree structure, it only affects the trees' ability to identify the risk factors when the sample size is small. With a reasonably large sample size, the proposed trees apparently identify the key factors by splitting on them, and their performance is resistant to high censoring rate. Also, larger sample sizes reduce the impact of high censoring on the trees' ability to recover the correct tree structure.   

Comparing the results in the two tables when $N=100$, one can find that the time-varying covariate $X_2$ is more frequently identified as a risk factor by the trees in the dichotomous type II than the dichotomous type I. This may due to the fact that time-varying covariate $X_2$ changes value more frequently in dichotomous type II, and therefore its effect is easier for the trees to capture. It also leads to higher recovery rate of the entire tree structure. 

It is clear that both proposed trees perform well with reasonably large sample size ($N \geq 300$).  LTRCART works better when the sample size is small and the censoring rate is low, while LTRCIT outperforms LTRCART when the sample size is large and the censoring rate is high. The advantage of LTRCIT in the high censoring case is consistent with the previous LTRC results. The advantage of LTRCART in the small sample situation can be explained by the fact that its proportional hazards assumption is satisfied in this simulation. Therefore, it performs better when the trees are more unstable due to small sample size.

\subsubsection{Continuous time-varying covariate}
So far, we have only considered the case where the time-varying split variable $X_2$ is binary. However, in practice time-varying covariates usually take on more possible values than two. Therefore, in this section we consider the case that the split variable $X_2$ is continuous.

To be specific, we randomly generate the time-varying covariate $X_2$ from $U(0,10)$ at each of the time points $\{t_0, t_1, t_2, t_3\}$. Everything else is kept the same as in the previous setup, including the generation of the other covariates, the parameters in Table \ref{TV-parameter} and the censoring structure. 

The true model here is 

\begin{equation}
h(t,x,z(t)) = h_0(t)e^{\beta x + \beta_z z(t)} = h_0(t)e^{\beta I_{\{X_1 = 1\}} + \beta_z I_{\{X_2 > 5\}}},
\end{equation}
with the corresponding underlying structure of data shown in Figure \ref{fig:TV-struct-ctn}. We run 1,000 simulation trials and report the percentage of times the correct tree structure is recovered by the proposed LTRC trees. Table \ref{tab-typeI-ctn} shows the results for the situation when $X_2$ changes once and Table \ref{tab-typeII-ctn} gives the result for the situation when $X_2$ changes multiple times.

\begin{figure}[t]
  \begin{center}
   \captionsetup{justification=centering}
      \includegraphics[width=6 in,height =3 in]{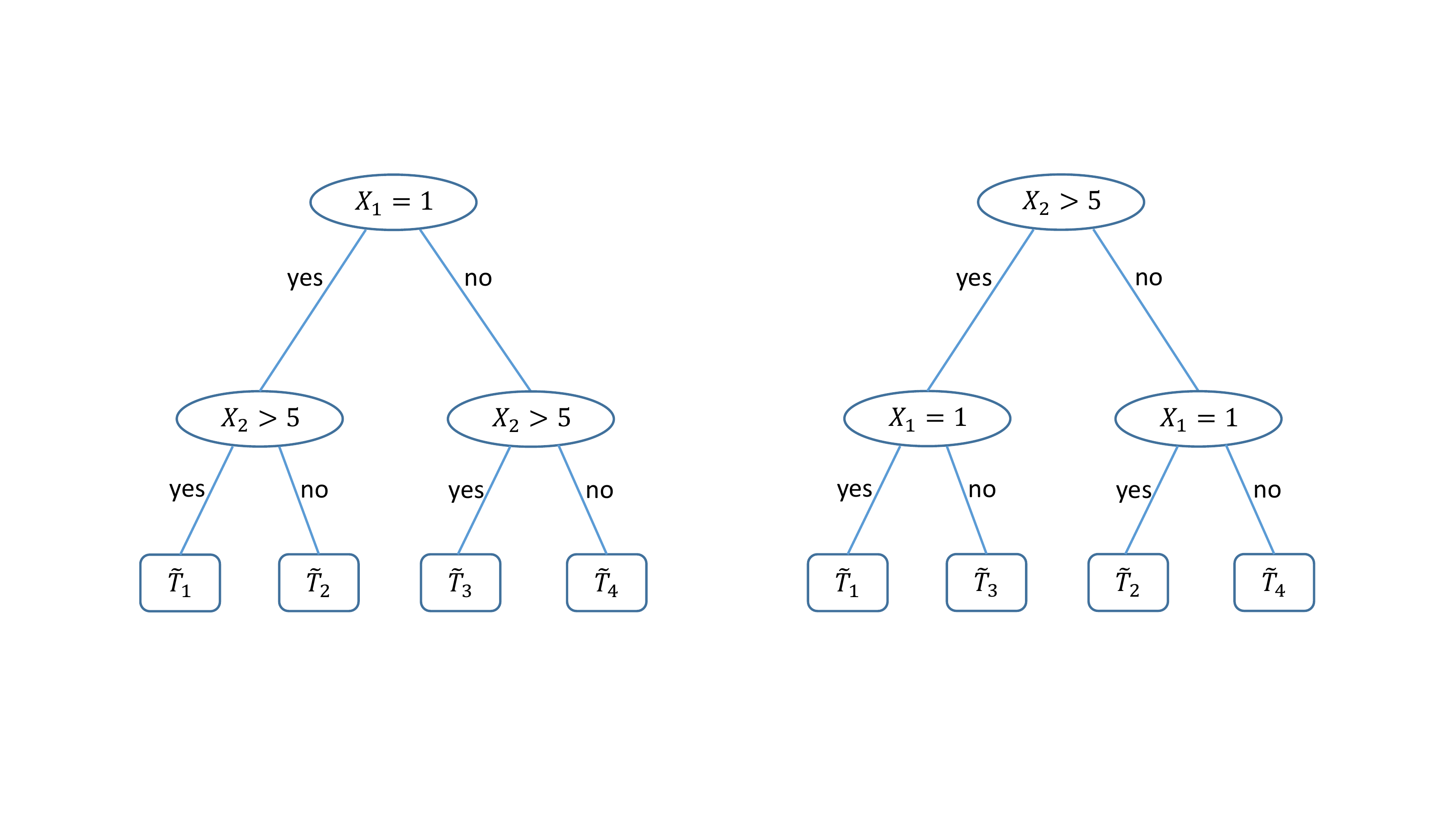}
   \end{center}
     \caption{\label{fig:TV-struct-ctn} Underlying data structure with continuous $X_2$} 
 \end{figure}

\begin{table}[t]
\center
\caption{\label{tab-typeI-ctn} Percentage of times correct tree structure recovered--Type I} 
\hspace*{-2cm}
\begin{threeparttable}
\scalebox{0.75}{%
\begin{tabular}{ccccccccccccccccccccccccc}
\hline
Censoring rate && \multicolumn{7}{c}{Exponential} && \multicolumn{7}{c}{Weibull} &&  \multicolumn{7}{c}{Gompertz} \\  \cline{3-9} \cline{11-17} \cline{19-25} 
$\%$ && \multicolumn{3}{c}{\footnotesize{LTRCIT}}&&\multicolumn{3}{c}{\footnotesize{LTRCART}} && \multicolumn{3}{c}{\footnotesize{LTRCIT}}&&\multicolumn{3}{c}{\footnotesize{LTRCART}} && \multicolumn{3}{c}{\footnotesize{LTRCIT}}&&\multicolumn{3}{c}{\footnotesize{LTRCART}} \\  \cline{3-5}  \cline{7-9}  \cline{11-13}  \cline{15-17}  \cline{19-21} \cline{23-25}
$N=100$ &&  $X_1$&$X_2$&S  &&  $X_1$&$X_2$&S && $X_1$&$X_2$&S  && $X_1$&$X_2$&S   &&  $X_1$&$X_2$&S && $X_1$&$X_2$&S \\ \hline

$0$  && $69.7$&$60.7$&$1.7$  &&  $52.1$&$73.5$&$5.2$  &&  $80.1$&$50.8$&$1.2$  && $67.8$&$86.2$&$7.7$   &&  $97.2$&$86.0$&$4.3$  &&  $94.6$&$99.9$&$7.1$ \\
$20$ && $58.4$&$55.8$&$1.1$  &&  $40.4$&$67.3$&$3.3$  &&  $71.2$&$38.5$&$0.6$  && $58.7$&$77.1$&$5.0$   &&  $94.6$&$70.2$&$1.0$  &&  $92.2$&$98.5$&$3.6$ \\
$50$ && $37.5$&$40.6$&$0.2$  &&  $24.3$&$52.2$&$1.1$  &&  $55.8$&$13.9$&$0.0$  && $43.6$&$46.9$&$0.4$   &&  $81.3$&$32.3$&$0.1$  &&  $73.5$&$75.3$&$0.7$ \\ \hline
$N=300$ &&  $$&$$&$$  &&  $$&$$&$$  &&  $$&$$&$$  && $$&$$&$$   &&  $$&$$&$$  &&  $$&$$&$$ \\ \hline
$0$  && $100$&$99.6$&$70.2$  &&  $99.2$&$99.9$&$56.7$  &&  $100$&$99.6$&$50.3$  && $99.7$&$100$&$59.9$ &&  $100$&$100$&$53.8$  &&  $100$&$100$&$66.3$ \\
$20$ && $99.6$&$99.2$&$59.2$ &&  $94.4$&$99.8$&$45.9$  &&  $100$&$96.8$&$35.2$  && $99.2$&$100$&$44.0$ &&  $100$&$100$&$39.6$  &&  $100$&$100$&$48.5$ \\ 
$50$ && $93.6$&$93.2$&$29.3$ &&  $74.6$&$95.6$&$21.7$  &&  $99.1$&$53.1$&$3.7$  && $94.9$&$91.6$&$16.4$&&  $100$&$90.7$&$16.6$  &&  $99.8$&$99.9$&$24.8$ \\  \hline
$N=500$ &&  $$&$$&$$  &&  $$&$$&$$  &&  $$&$$&$$  && $$&$$&$$   &&  $$&$$&$$  &&  $$&$$&$$ \\ \hline
$0$  && $100$&$100$&$81.2$  && $100$&$100$&$82.6$  &&  $100$&$100$&$78.9$  && $100$&$100$&$83.8$ && $100$&$100$&$81.5$ && $100$&$100$&$85.9$ \\
$20$ && $100$&$100$&$81.9$  && $100$&$100$&$78.3$  &&  $100$&$99.9$&$70.5$ && $100$&$100$&$70.5$ && $100$&$100$&$73.1$ && $100$&$100$&$80.6$ \\ 
$50$ && $99.9$&$99.7$&$68.7$&& $95.7$&$99.8$&$51.1$&&  $100$&$87.6$&$24.3$ && $99.9$&$99.2$&$45.6$&&$100$&$100$&$46.3$ && $100$&$100$&$49.7$ \\ \hline 
\end{tabular} 
\hspace*{-2cm}
} \\
 \begin{tablenotes}[para,flushleft]
\footnotesize {``S'' means correct tree structure as in Figure \ref{fig:TV-struct-ctn}}.
  \end{tablenotes}
\end{threeparttable}
\end{table}

\begin{table}[t]
\center
\caption{\label{tab-typeII-ctn} Percentage of times correct tree structure recovered--Type II} 
\begin{threeparttable}
\scalebox{0.75}{%
\hspace*{-2cm}
\begin{tabular}{ccccccccccccccccccccccccc}
\hline
Censoring rate && \multicolumn{7}{c}{Exponential} && \multicolumn{7}{c}{Weibull} &&  \multicolumn{7}{c}{Gompertz} \\  \cline{3-9} \cline{11-17} \cline{19-25} 
$\%$ && \multicolumn{3}{c}{\footnotesize{LTRCIT}}&&\multicolumn{3}{c}{\footnotesize{LTRCART}} && \multicolumn{3}{c}{\footnotesize{LTRCIT}}&&\multicolumn{3}{c}{\footnotesize{LTRCART}} && \multicolumn{3}{c}{\footnotesize{LTRCIT}}&&\multicolumn{3}{c}{\footnotesize{LTRCART}} \\  \cline{3-5}  \cline{7-9}  \cline{11-13}  \cline{15-17}  \cline{19-21} \cline{23-25}
$N=100$ &&  $X_1$&$X_2$&S  &&  $X_1$&$X_2$&S && $X_1$&$X_2$&S  && $X_1$&$X_2$&S   &&  $X_1$&$X_2$&S && $X_1$&$X_2$&S \\ \hline

$0$  && $73.4$&$82.5$&$5.6$  &&  $59.2$&$89.7$&$7.6$  &&  $79.7$&$81.4$&$6.8$ && $66.4$&$95.9$&$11.2$&&  $97.3$&$97.7$&$25.1$  &&  $91.3$&$100$&$24.8$ \\
$20$ && $61.4$&$74.5$&$2.8$  &&  $44.9$&$81.9$&$4.9$  &&  $71.6$&$63.8$&$3.0$ && $56.3$&$87.1$&$7.4$ &&  $93.8$&$91.3$&$13.5$  &&  $85.2$&$99.6$&$13.2$ \\
$50$ && $43.2$&$57.1$&$0.3$  &&  $27.6$&$64.2$&$2.0$  &&  $56.0$&$23.2$&$0.3$ && $40.3$&$56.2$&$1.4$ &&  $74.9$&$52.7$&$0.9$  &&  $64.6$&$86.3$&$3.1$ \\ \hline
$N=300$ &&  $$&$$&$$  &&  $$&$$&$$  &&  $$&$$&$$  && $$&$$&$$   &&  $$&$$&$$  &&  $$&$$&$$ \\ \hline
$0$  && $100$&$100$&$76.5$  &&  $99.8$&$100$&$51.2$  && $99.9$&$100$&$78.1$  && $99.6$&$100$&$75.4$ && $100$&$100$&$85.8$  &&  $100$&$100$&$87.6$ \\
$20$ && $99.9$&$100$&$69.5$ &&  $96.5$&$100$&$42.7$  && $100$&$99.6$&$66.0$  && $99.4$&$100$&$56.2$ && $100$&$100$&$81.8$  &&  $100$&$100$&$78.3$ \\ 
$50$ && $94.5$&$98.6$&$40.9$&&  $77.5$&$99.3$&$25.6$ && $99.2$&$82.1$&$19.5$ && $91.9$&$97.9$&$23.1$&& $100$&$99.4$&$52.8$  &&  $99.7$&$100$&$42.9$ \\  \hline
$N=500$ &&  $$&$$&$$  &&  $$&$$&$$  &&  $$&$$&$$  && $$&$$&$$   &&  $$&$$&$$  &&  $$&$$&$$ \\ \hline
$0$  && $100$&$100$&$84.8$ && $100$&$100$&$81.8$  &&  $100$&$100$&$88.4$  && $100$&$100$&$92.9$   &&  $100$&$100$&$85.4$  &&  $100$&$100$&$92.4$ \\
$20$ && $100$&$100$&$83.0$ && $99.9$&$100$&$74.2$ &&  $100$&$100$&$82.0$  && $100$&$100$&$83.7$   &&  $100$&$100$&$82.5$  &&  $100$&$100$&$89.3$ \\ 
$50$ && $99.9$&$100$&$78.2$&& $97.3$&$100$&$61.1$ &&  $99.9$&$98.0$&$57.9$&& $99.4$&$99.6$&$49.0$ &&  $100$&$100$&$82.1$  &&  $100$&$100$&$74.6$ \\ \hline 
\end{tabular} 
\hspace*{-2cm}
} \\
\begin{tablenotes}[para,flushleft]
\footnotesize {``S'' means correct tree structure as in Figure \ref{fig:TV-struct-ctn}}.
 \end{tablenotes}
\end{threeparttable}

\end{table}

Comparing to the results in Table \ref{tab-typeI} and Table \ref{tab-typeII}, we can see that the trees recover the correct tree structure less often when the time-varying covariate $X_2$ is continuous compared to when it is binary. This is understandable, since the binary variable gives a clear binary cut for splitting, and hence the correct split is easier to be identified by the tree than for the continuous variable. Nevertheless, both proposed trees still perform well given reasonably large sample size.

As is true in the binary $X_2$ case, we can see the proposed trees perform better when $X_2$ can change multiple times compared to the setting where $X_2$ can change only once. This is especially helpful in the continuous time-varying covariate situation where more observations are needed for the tree to identify the correct split than in the binary time-varying covariate case.  

\subsection{Prediction performance}
We test the prediction performance of the proposed methods in this section using the integrated Brier score (IBS). The training set is generated as in Section 6.1. From Figure \ref{fig-TV-struct} and Figure \ref{fig:TV-struct-ctn}, we can see that there are total four distinct survival distributions as denoted by the four terminal nodes in Figure \ref{fig-TV-struct} and Figure \ref{fig:TV-struct-ctn}. Therefore, we generate the test set from the four survival distributions, and compare it with the predicted survival distribution from the fitted trees.

More specifically, the test set contains five covariates whose values are generated according to the description in Section 6.1. The survival time $T$ is generated according to each distribution as follows: 
\begin{itemize}
	\item Exponential \hspace{0.05in} $T = -\frac{\log(u)}{\lambda \exp(\vartheta)}$
	\item Weibull     \hspace{0.1in} $T = \left(-\frac{\log(u)}{\lambda \exp(\vartheta)}\right)^{1/\nu}$
	\item Gompertz    \hspace{0.05in} $T =\frac{1}{\alpha} \log \left(1- \frac{\alpha \log(u)}{\lambda \exp(\vartheta)}\right)$
\end{itemize}
where $u \sim$ $U(0,1)$. The value of $\vartheta$ is determined by the corresponding terminal node in Figure \ref{fig-TV-struct} and Figure \ref{fig:TV-struct-ctn}:

\begin{itemize}
	\item Node $\widetilde{T}_1$: $\vartheta = \beta + \beta_z$
	\item Node $\widetilde{T}_2$: $\vartheta = \beta$
	\item Node $\widetilde{T}_3$: $\vartheta = \beta_z$
	\item Node $\widetilde{T}_4$: $\vartheta = 0$
\end{itemize}

Note that the covariates in the test set are time-independent, instead of being time-varying as in the training set. Nevertheless, they contain the same four survival distributions. All of the settings in Section 6.1, binary and continuous time-varying covariate $X_2$, coupled with $X_2$ changes value only once and multiple times, are tested in this section. Note that for node $\widetilde{T}_2$ and $\widetilde{T}_4$ in Figure \ref{fig-TV-struct} and \ref{fig:TV-struct-ctn}, the survival time $T$ is right-censored at time $6$,  while for node $\widetilde{T}_1$ and $\widetilde{T}_3$ the survival time is left truncated at time $0.6$. This is because $t_0$ is generated from $U(0.6,6)$. 

For each setting, the sample size $N$ in the test set is set equal to the sample size in the training set, which takes values from $\{100, 300, 500\}$. Five hundred simulation runs are performed for each setting and the resulting IBS values are shown in the boxplots. The two proposed trees are compared to each other, as well as to the Cox proportional hazards model. The signed-rank test is used in each setting to compare the methods. Figures \ref{fig:binary} and \ref{fig:continuous} give side-by-side boxplots with binary and continuous $X_2$, respectively, with $X_2$ changing value multiple times and sample size $N=300$. In each figure, the first row represents the $0\%$ censoring rate case, the second row represents the $20\%$ censoring rate case and the third row represents the $50\%$ censoring rate case. Complete results can be found in the supplemental material. The results are as follows.

\begin{itemize}
	\item \textbf{Binary time-varying covariate $X_2$} \\
	Note that both the tree model and the proportional hazards assumption are satisfied when $X_2$ is binary (Figure \ref{fig-TV-struct} and equation $(9)$). The results show that both trees perform significantly better than the Cox model in the Weibull and Gompertz distribution cases, while the Cox model performs better for the Exponential distribution. This is true regardless of censoring rate and sample size. LTRCIT generally performs better than LTRCART. 
	
	Since both the tree and the Cox model are the true model in this case, these results indicate that the proposed trees have comparable prediction performance compared to the Cox model. 
	     
	\item \textbf{Continuous time-varying covariate $X_2$} \\
	The results in the continuous time-varying covariate $X_2$ case are broadly similar to those in the binary case. The most notable difference is that trees have less of an advantage over the Cox model for the Weibull and Gompertz distributions, which comes as a surprise since the Cox model no longer holds (the log hazard is not linearly dependent on $X_2$, while the proportional hazards assumption still hold). The reason is that trees recover the correct tree structure less often in the continuous $X_2$ case than in the binary case, and therefore its prediction performance is undermined. 
	
	In contrast, the performance of the Cox model is not affected by the continuous $X_2$. This may due to that the effect of $X_2$ on the hazard is simple and monotonic, so the Cox model can approximate the data well. Note that the proportional hazards assumption still holds here. Nevertheless, trees still perform significantly better than the Cox model with reasonably large sample size. 
\end{itemize}

As usual, trees work better on data with large sample size than on data with small sample size. Comparing the results for $N=500$ to those for $N=100$, one can easily see that trees have relatively better performance when the sample size is large. Because trees are relatively unstable compared to the Cox model, prediction accuracy is dominated by larger variance when the sample size is small. However, when the sample size is large, stability is no longer an issue, in which case trees become more favorable because they are more flexible.

\begin{sidewaysfigure}
\includegraphics[width=\textwidth]{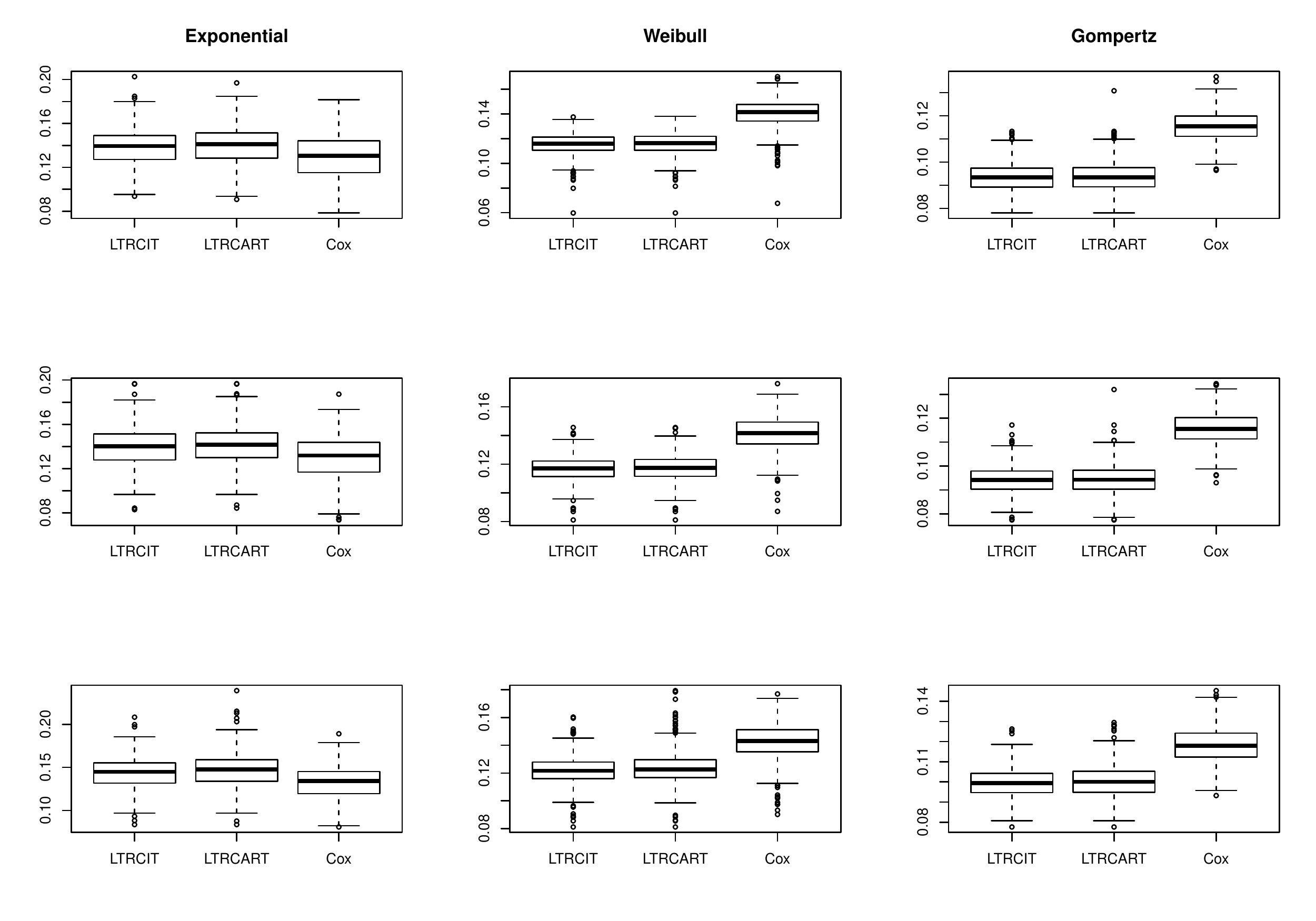}
\caption{ \label{fig:binary} IBS boxplots with binary $X_2$, dichotomous type II and $N=300$. The first to the third row represents $0\%$,  $20\%$ and $50\%$ censoring rate, respectively.}
\end{sidewaysfigure}

\begin{sidewaysfigure}
\includegraphics[width=\textwidth]{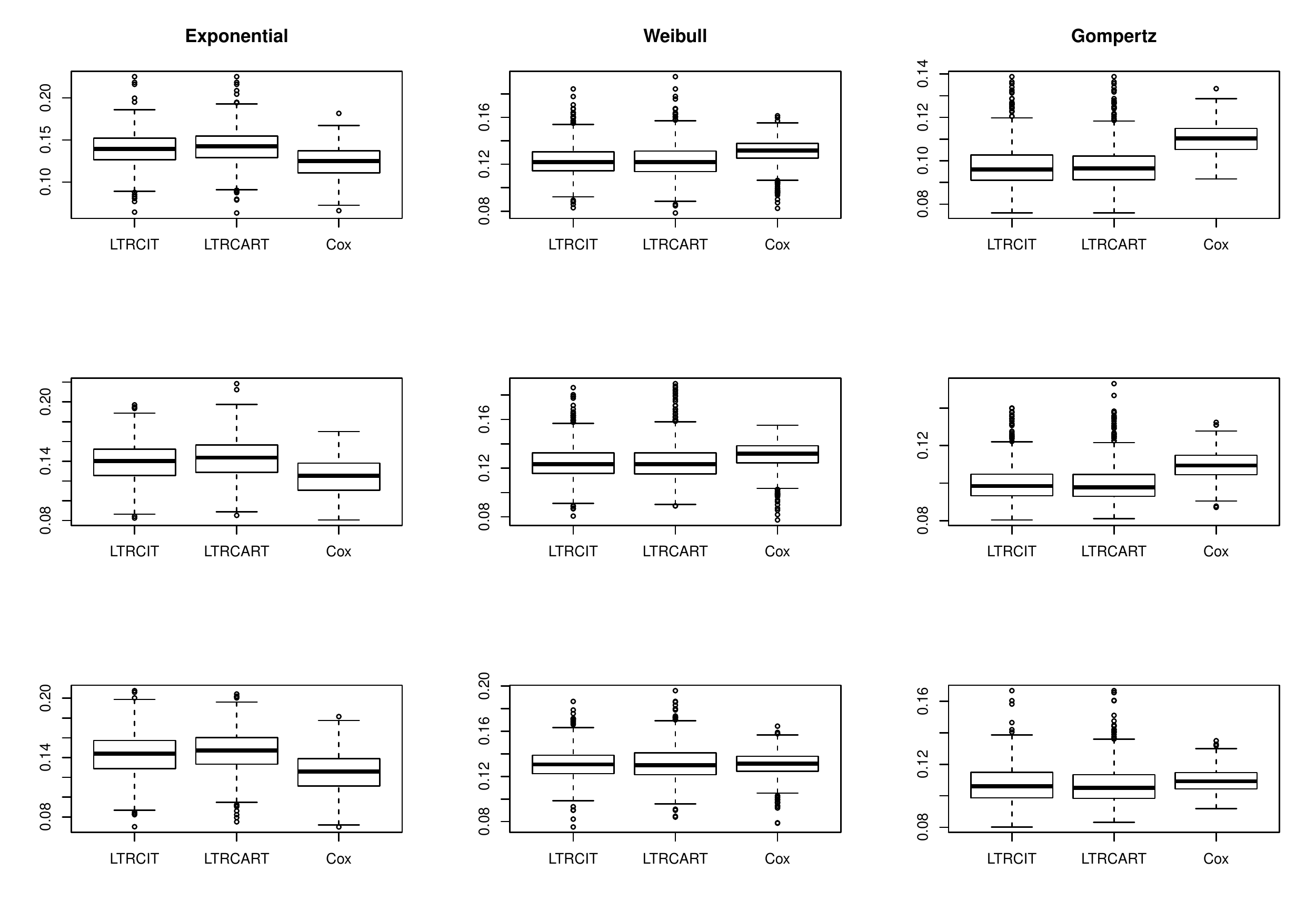}
\caption{\label{fig:continuous}IBS boxplots with continuous $X_2$ whose value changes multiple times and $N=300$. The first to the third row represents $0\%$,  $20\%$ and $50\%$ censoring rate, respectively.}
\end{sidewaysfigure}

\section{Real data applications}
We will test the proposed time-varying covariates trees on two real data examples. 

\subsection{Bone Marrow Transplants Data}
	The Bone Marrow Transplants Data in the \texttt{R} package \texttt{KMsurv} is described in \cite{klein2003survival}. 
Bone marrow transplants (BMT) are a standard treatment for acute leukemia.
The interest is in examining the relationship between disease-free survival time after the transplantation and a set 
of factors for patients given a bone marrow transplant. Besides those time-independent 
covariates measured at time of transplant, such as disease group and the French-American-British (FAB)
classification based on standard morphological criteria, there are three intermediate events that
occur during the transplant recovery process that may affect the disease-free survival
time of a patient. These are the development of acute graft-versus-host
disease (aGVHD), the development of chronic graft-versus-host disease (cGVHD) and the 
return of the patient's platelet count to a self-sustaining level. They serve as the 
time-varying covariates in our analysis. The risk factors considered are thus:
\begin{itemize}
	\item Group: Disease Group $1$-ALL, $2$-AML Low Risk, $3$-AML High Risk
	\item FAB: $1$-FAB grade $4$ or $5$, $0$-Otherwise
	\item aGVHD: $1$-Developed acute GVHD, $0$-Otherwise
	\item cGVHD: $1$-Developed chronic GVHD, $0$-Otherwise
	\item Platelet: $1$-Platelets returned to normal, $0$-Otherwise
\end{itemize}

\begin{figure}[t]
\centering
\begin{minipage}{\linewidth}
  \begin{minipage}{0.45\linewidth}
    \includegraphics[width=\linewidth]{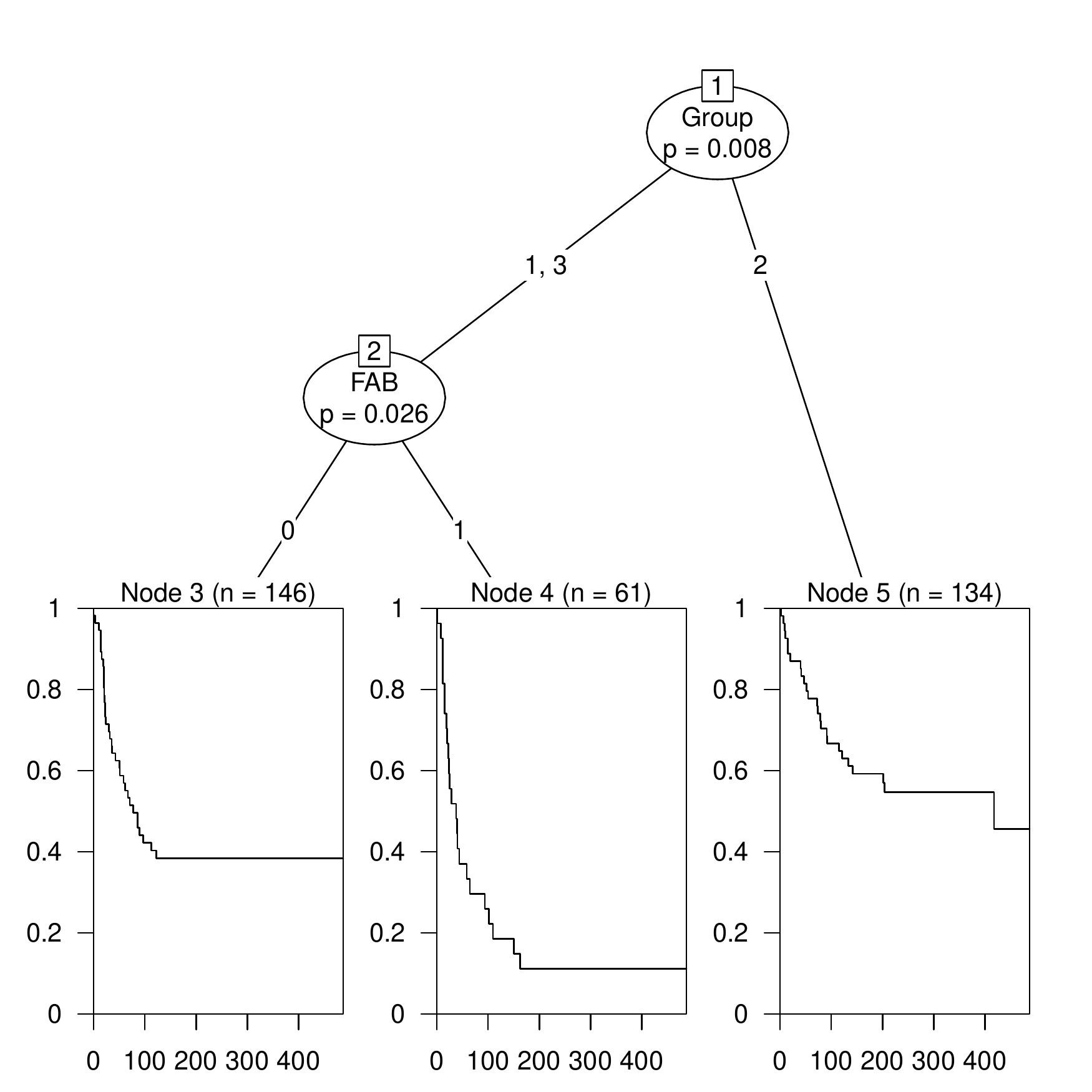}
  \end{minipage}
  \hspace{0.05in}
  \begin{minipage}{0.45\linewidth}
    \includegraphics[width=\linewidth]{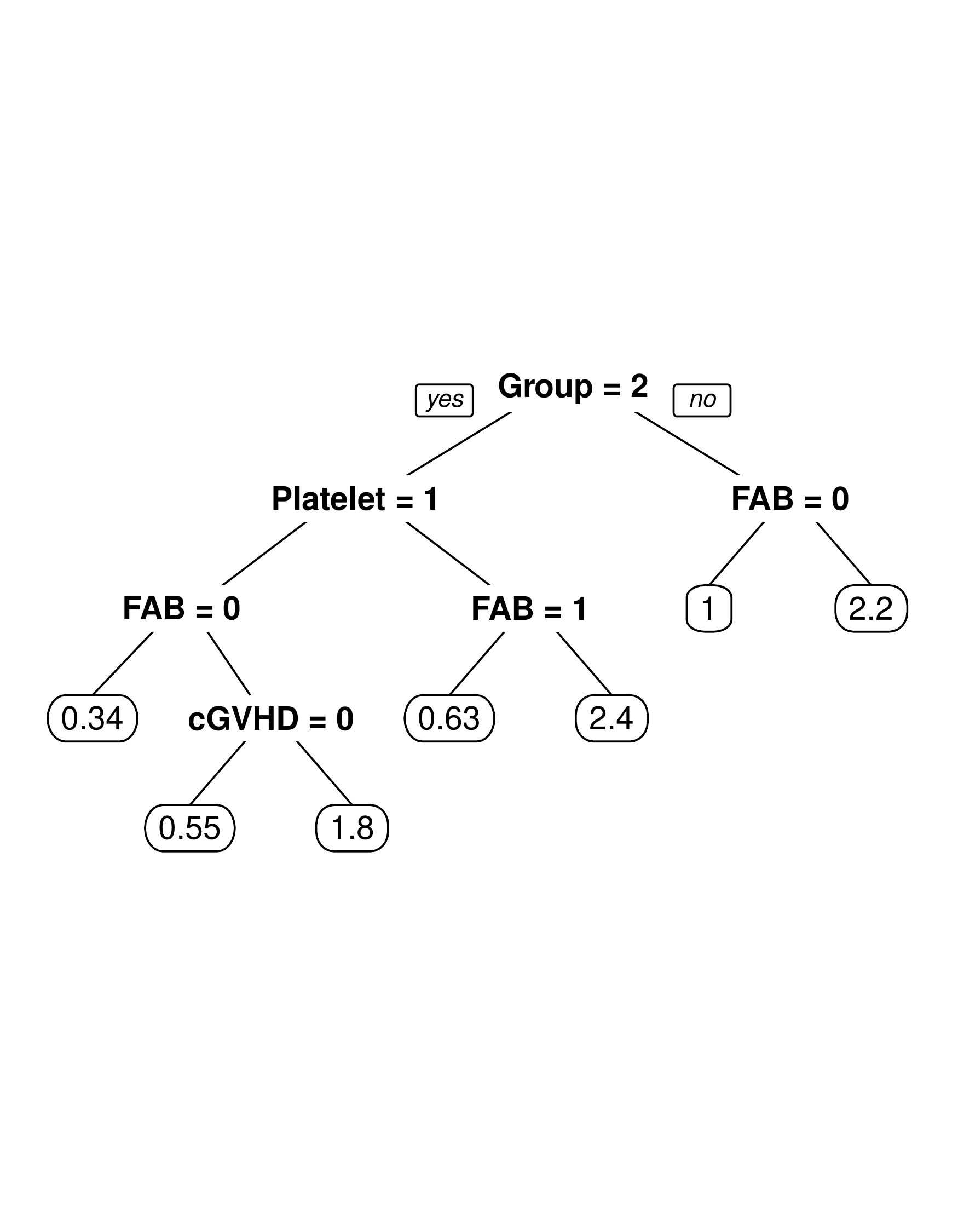}
  \end{minipage}
\end{minipage}
\caption{\label{fig:BMT}Left and right plots are LTRCIT and LTRCART trees for the BMT data, respectively.}
\end{figure}

 \begin{table}[t]
\center
\caption{\label{cox:BMT}  Cox model result of BMT data} 
\begin{tabular}{ lcrcccccrcc}
\hline
 Risk factor && $coef$ && $exp(coef)$ && $se(coef)$ && $z$ && $p$\\
  \hline
   aGVHD   && $0.20833$ && $1.232$   && $0.296$ && $0.703$ && $0.4800$\\    
   cGVHD   && $-0.16750$ && $0.846$  && $0.290$ && $-0.578$ && $0.5600$\\    
 Platelet  && $-0.95407$ && $0.385$  && $0.337$ && $-2.834$ && $0.0046$\\  
 FAB       && $0.69162$ && $1.997 $  && $0.274$ && $2.523$ && $0.0120$\\  
Group-$2$  && $-0.81616$ && $0.442$  && $0.329$ && $-2.480$ && $0.0130$\\  
Group-$3$  && $0.00531$ && $1.005$   && $0.331$ && $0.016$ && $0.9900$\\  
  \hline
\end{tabular}
\end{table}

Figure \ref{fig:BMT} shows the results of the two proposed LTRC trees applied to the BMT data, while Table \ref{cox:BMT} contains the corresponding result from the Cox proportional hazards model. Since both time-independent covariates, disease group and FAB, are identified as important risk factors by the Cox model and the two tree algorithms, one can see their main difference lies on the time-varying covariates. Among the three time-varying covariates, only Platelet is identified as an important risk factor by the Cox model. In contrast, none of the three time-varying covariates are split variables for LTRCIT, while both Platelet and cGVHD are considered as important risk factors by LTRCART. 

The two tree algorithms agree on the branch where disease group is ALL or AML High Risk, but diverge on disease group AML Low Risk. LTRCART splits on Platelet, FAB and cGVHD on that branch while LTRCIT does not split at all. Table \ref{cox:BMT2} shows the Cox model result on that branch. One can see it also picks Platelet and FAB as predictive risk factors. Tests of the proportional hazards assumption based on the Schoenfeld residuals shows that this assumption is reasonable for these data. From these results, it is clear that compared to LTRCIT, LTRCART gives results that are more similar to those of the Cox model, which should not be a surprise given that both LTRCART and the Cox model are based on the assumption of proportional hazards, especially since in this case the assumption apparently holds.

  \begin{table}[t]
\center
\caption{\label{cox:BMT2}  Cox model result on subset of BMT data, with disease group-AML Low Risk} 
\begin{tabular}{ lcrcccccrcl}
\hline
 Risk factor && $coef$ && $exp(coef)$ && $se(coef)$ && $z$ && $p$\\
  \hline
   aGVHD   && $-0.0415$ && $0.9593$  && $0.555$ && $-0.0748$ && $0.9400$\\     
   cGVHD   && $0.2060$  && $1.2287$  && $0.502$ && $0.4099$  && $0.6800$\\  
 Platelet  && $-3.2893$ && $0.0373$  && $0.795$ && $-4.1384$ && $3.5\times 10^{-05}$\\  
 FAB       && $0.9191$  && $2.5070$  && $0.441$ && $2.0839$  && $3.7\times 10^{-02}$\\  
  \hline
\end{tabular}
\end{table}

\subsection{Mayo Clinic Primary Biliary Cirrhosis Data} 
This data set in the \texttt{R} package \texttt{survival} were obtained from 312 patients with primary biliary cirrhosis (PBC) enrolled in a double-blind, placebo-controlled, randomized trial conducted between January, 1974 and May, 1984 at the Mayo Clinic to evaluate the use of D-penicillamine for treating PBC. The data were collected at entry and at yearly intervals on a total of 45 variables. More detailed description can be found in \cite{dickson1989prognosis}. Follow-up was extended to April, 1988, which generated 1,945 patient visits that enable us to study the change in the prognostic variables of PBC \citep{murtaugh1994primary}.    
  
\cite{dickson1989prognosis} developed a predictive survival model based on the baseline data (time invariant data collected at entry). They used the Cox proportional hazards model, coupled with forward and backward stepwise variable selection procedures to build the model. Twelve noninvasive, easily collected variables that require only clinical evaluation and a blood sample were included in the modeling. These variables are as follows:
\begin{itemize}
	\item age: 	in years
	\item albumin:	logarithm of serum albumin (g/dl)
	\item alk.phos:	alkaline phosphotase (U/liter)
	\item ascites:	presence of ascites
	\item ast: aspartate aminotransferase(U/ml)
	\item bili:	logarithm of serum bilirubin (mg/dl)
	\item chol:	serum cholesterol (mg/dl)
	\item edema: $0$-no edema, $0.5$-untreated or successfully treated, $1$-edema despite diuretic therapy
	\item hepato: presence of hepatomegaly or enlarged liver
	\item platelet:	platelet count
	\item protime: logarithm of prothrombin time, standardized blood clotting time
	\item spiders:	presence or absence of spiders
\end{itemize}
The forward and backward stepwise selection procedures chose the same model, which contains five variables: age, edema, bili, albumin and protime. Indeed, if we run the Cox model on these $12$ variables, only those five variables have $p$-values less than $0.05$ (top panel in Table \ref{cox:PBC}). The survival tree results for these data are shown in Figure \ref{fig:PBC}. The conditional inference survival tree identifies the same five risk factors as the Cox model, while the relative risk survival tree identifies a different five risk factors: age, alk.phos, ascites, bili and protime. The main difference between the two trees is their left branches, where the conditional inference tree only splits on edema while the relative risk tree splits on age, alk.phos and protime.
  \begin{table}[p]
\center
\caption{\label{cox:PBC}  Cox model results of PBC data} 
\begin{threeparttable}
\begin{tabular}{ lrrcrc}
\hline
 & \multicolumn{5}{c}{Baseline result} \\ \cline{2-6} 
 Risk factor & \multicolumn{1}{c}{$coef$}  & $exp(coef)$ & $se(coef)$ & $z$ & $p$\\ \hline
 age    & $3.52\times 10^{-02}$ & $1.04$  & $1.02\times 10^{-02}$ & $3.47$ & $5.3\times 10^{-04}$\\     
albumin & $-2.09\times 10^{-00}$& $0.12$ & $9.06\times 10^{-01}$ & $-2.31$ & $2.1\times 10^{-02}$\\    
alk.phos& $1.12\times 10^{-05}$ & $1.00$  & $3.74\times 10^{-05}$ & $0.30$ & $7.6\times 10^{-01}$\\    
ascites & $4.61\times 10^{-01}$ & $1.59$  & $3.35\times 10^{-01}$ & $1.38$ & $1.7\times 10^{-01}$\\    
 ast	& $3.25\times 10^{-03}$ & $1.00$  & $1.89\times 10^{-03}$ & $1.73$ & $8.4\times 10^{-02}$\\  
 bili 	& $7.17\times 10^{-01}$ & $2.05$  & $1.41\times 10^{-01}$ & $5.08$ & $3.7\times 10^{-07}$\\   
 chol 	& $-6.50\times 10^{-05}$& $1.00$ & $4.54\times 10^{-04}$ & $-0.14$ & $8.9\times 10^{-01}$\\ 
 edema  & $7.44\times 10^{-01}$ & $2.10$  & $3.60\times 10^{-01}$ & $2.07$ & $3.9\times 10^{-02}$\\  
 hepato & $1.54\times 10^{-01}$ & $1.17$  & $2.34\times 10^{-01}$ & $0.66$ & $5.1\times 10^{-01}$\\   
platelet& $1.18\times 10^{-04}$ & $1.00$  & $1.13\times 10^{-03}$ & $0.10$ & $9.2\times 10^{-01}$\\   
 protime& $2.74\times 10^{-00}$ & $15.45$ & $1.15\times 10^{-00}$ & $2.37$ & $1.8\times 10^{-02}$\\   
 spiders& $1.62\times 10^{-01}$ & $1.18$  & $2.30\times 10^{-01}$ & $0.71$ & $4.8\times 10^{-01}$\\   
  \hline
 &\multicolumn{5}{c}{Entire follow-up result} \\ \cline{2-6} 
 Risk factor & \multicolumn{1}{c}{$coef$}  & $exp(coef)$ & $se(coef)$ & $z$ & $p$ \\ \hline
 age & $4.01\times 10^{-02}$ & $1.04$ & $1.5\times 10^{-02}$  &  $2.61$ & $9.1\times 10^{-03}$ \\
 albumin & $-3.24\times 10^{-00}$& $0.04$ & $9.2\times 10^{-01}$  &  $-3.53$ & $4.2\times 10^{-04}$\\
 alk.phos & $-9.10\times 10^{-05}$& $1.00$ & $1.8\times 10^{-04}$  &  $-0.50$ & $6.2\times 10^{-01}$ \\
 ascites & $3.02\times 10^{-01}$ & $1.35$ & $3.7\times 10^{-01}$  &  $0.81$ & $4.2\times 10^{-01}$ \\
 ast  & $-1.92\times 10^{-03}$& $0.99$ & $2.3\times 10^{-03}$  &  $-0.84$ & $4.0\times 10^{-01}$ \\
 bili & $8.69\times 10^{-01}$ & $2.38$ & $2.1\times 10^{-01}$  &  $4.15$ & $3.4\times 10^{-05}$ \\
 chol& $-6.59\times 10^{-04}$& $1.00$ & $1.0\times 10^{-03}$  &  $-0.64$ & $5.2\times 10^{-01}$ \\
 edema & $7.17\times 10^{-01}$ & $2.05$ & $4.8\times 10^{-01}$  &  $1.50$ & $1.3\times 10^{-01}$ \\
 hepato & $-5.22\times 10^{-01}$& $0.59$ & $3.7\times 10^{-01}$  &  $-1.42$ & $1.6\times 10^{-01}$ \\
 platelet & $2.27\times 10^{-03}$ & $1.00$ & $1.6\times 10^{-03}$  &  $1.41$ & $1.6\times 10^{-01}$ \\
 protime & $3.32\times 10^{-00}$ & $27.60$& $1.1\times 10^{-00}$  &  $2.90$ & $3.8\times 10^{-03}$ \\
 spiders & $3.71\times 10^{-01}$ & $1.45$ & $3.4\times 10^{-01}$  &  $1.09$ & $2.8\times 10^{-01}$ \\ \hline
\end{tabular}
 \begin{tablenotes}[para,flushleft]
\footnotesize{The ``baseline result'' is obtained using only baseline (time invariant) data; while the \\``follow-up result'' is obtained using all (time varying) data}.
  \end{tablenotes}
\end{threeparttable}
\end{table}

\begin{figure}[t]
\centering
\begin{minipage}{\linewidth}
  \begin{minipage}{0.45\linewidth}
    \includegraphics[width=\linewidth]{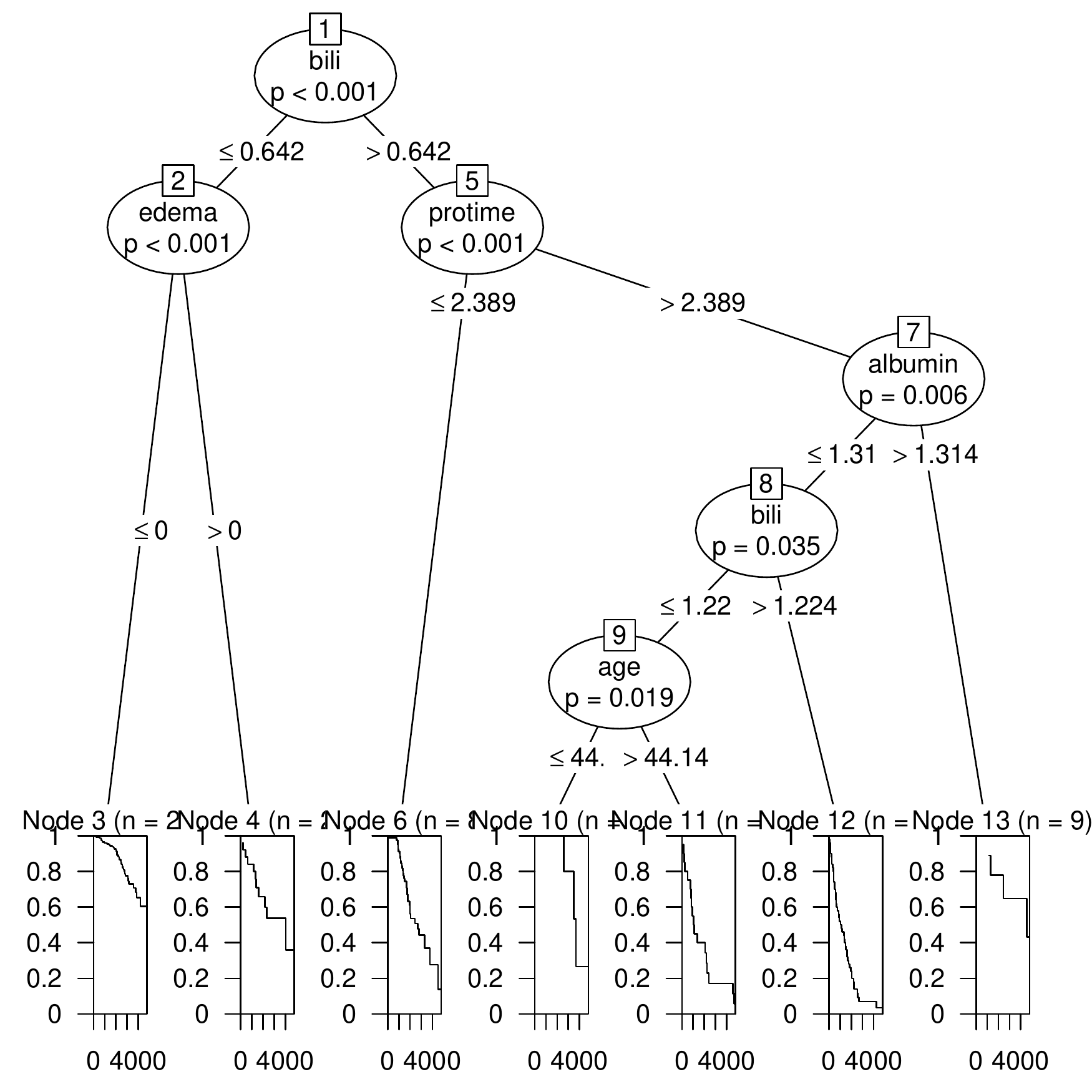}
  \end{minipage}
  \hspace{0.05in}
  \begin{minipage}{0.45\linewidth}
    \includegraphics[width=\linewidth]{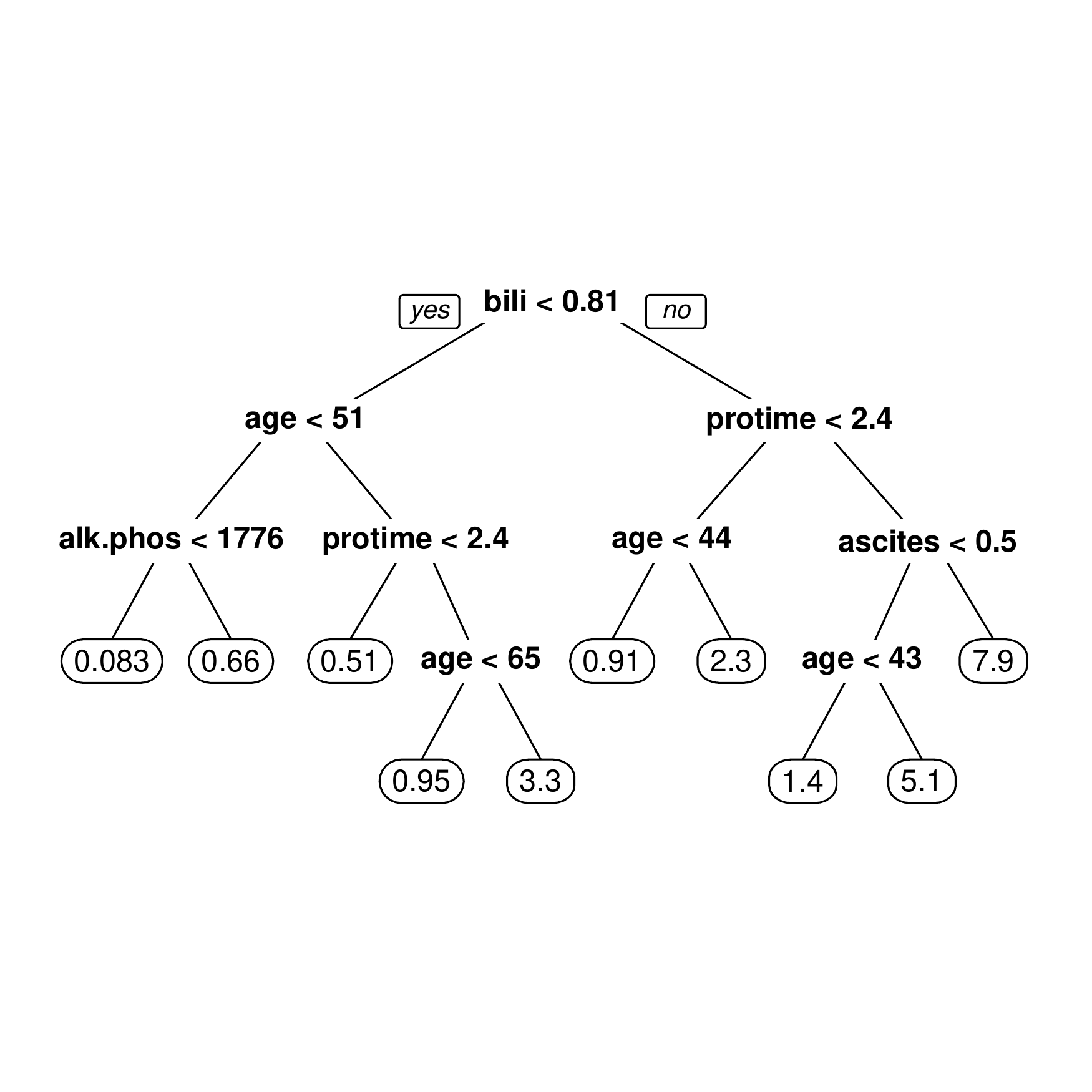}
  \end{minipage}
\end{minipage}
\caption{Left and right panels show the conditional inference survival tree and relative risk survival tree, respectively, for PBC data based on baseline (time invariant) variables}
\label{fig:PBC}
\end{figure}

However, it is interesting to see whether the results change when the entire follow-up data ($1,945$ LTRC observations) are used in the analysis. Note that all of the $12$ covariates except age become time-varying covariates in the follow-up data. The Cox model result on the entire follow-up data is shown on the lower panel in Table \ref{cox:PBC}, from which we can see that age, bili, albumin and protime are still significant at a $0.05$ level while edema is no longer significant. Among those variables that are statistically significant, the effects of albumin and prothrombin time seem larger in the time-varying analysis than in the analysis using only baseline values. The other variables remain insignificant when the entire follow-up data are used in the analysis.   

We also fit the two time-varying covariates survival trees on the follow-up data and show the results in Figure \ref{fig:PBC2}. The LTRCIT tree also identifies age, bili, albumin and protime as important risk factors, along with ascites and spiders. Note that the last two variables are not considered as important risk factors in the  corresponding time-independent version of the survival tree fitted using baseline data. Just as was true for the Cox model, edema is also dropped as an important risk factor in the time-varying covariate results.  It is striking that
the top-level split variables are different when the follow-up data are included in the analysis.   

LTRCART selects the five covariates age, edema,  bili, albumin and protime as important risk factors, while its time-independent version does not pick albumin as a risk factor. This shows that although it shares the proportional assumption with the Cox model, its result does not necessary correspond to that of the Cox model. 

In practice, it is best for researchers to analyze data using several different models/tools before making any decision or drawing any conclusion. Since every statistical model/tool has its own assumptions or applicable conditions, it is risky for decision-making to rely on the result of just one model. Rather, the common results of different models/tools usually represent a summary of the actual information in the data, and is thus more reliable. In the PBC data case, it is safe to say that age, bilirubin, albumin and prothrombin time are predictive risk factors for survival time of individuals with primary biliary cirrhosis, since they are identified by all of the models, while the importance of other potential risk factors such as edema may be decided by further analysis or domain expertise. 
 
\begin{figure}[H]
\centering
\begin{minipage}{\linewidth} \centering
    \includegraphics[width=0.8\linewidth, height=0.4\textheight]{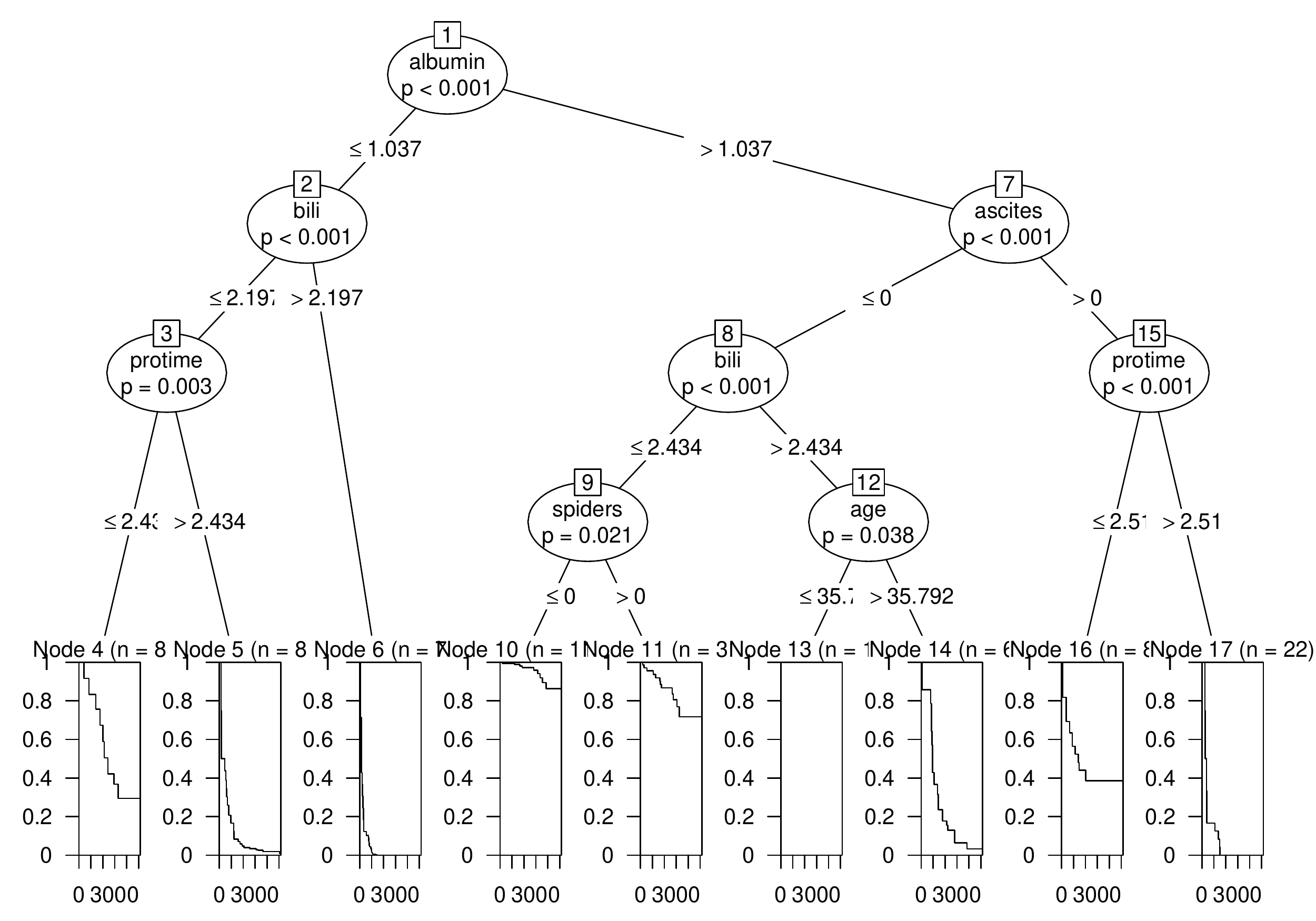}
\end{minipage}
\begin{minipage}{\linewidth} \centering
    \includegraphics[width=0.8\linewidth, height=0.4\textheight]{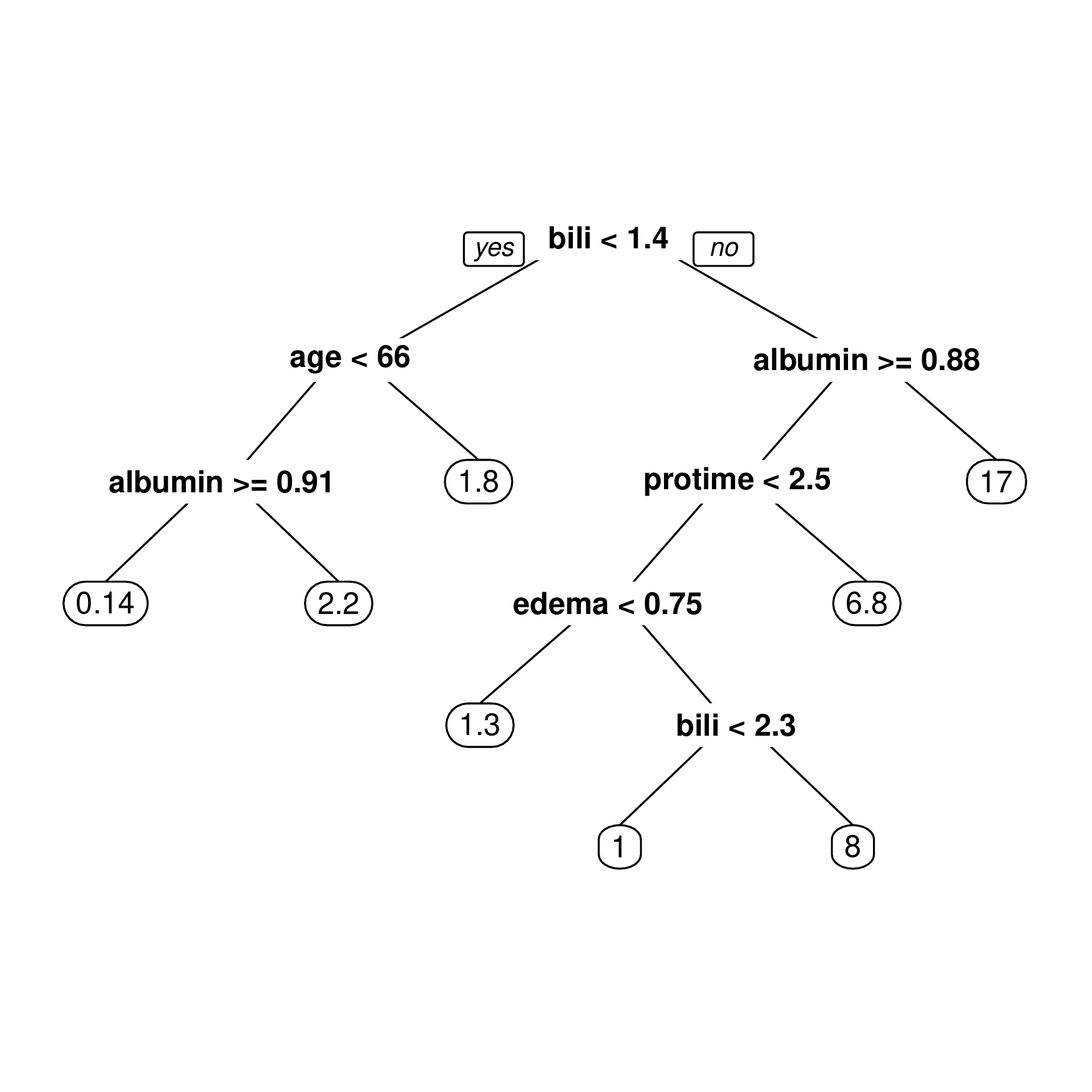}
\end{minipage}
\caption{Upper and lower panels show the LTRCIT and LTRCART trees for PBC data, respectively}
\label{fig:PBC2}
\end{figure}

\section{Conclusion}
In this paper we have proposed two left-truncation and right-censored (LTRC) tree methods. They are each an extension of an existing survival tree algorithm. Simulations are used to explore the properties of the proposed LTRC trees, including the unbiasedness of the tree algorithms, the trees' ability to recover the correct tree structure, and their prediction performance. Results show that with a reasonably large sample size, both LTRC trees perform well in terms of recovering true underlying structure of data and their prediction performance compares well with the Cox proportional hazards model. Both trees are applied to a real data example and the results indicate that trees provide a good alternative to the Cox model and have several advantages over this (semi-)parametric model.

We also showed that the proposed LTRC trees can be used to fit time-varying covariate survival trees, and showed that transforming subjects with time-varying covariates into pseudo-subjects that are LTRC data with time-independent covariates is theoretically justifiable in tree construction. The time-varying covariate survival trees' ability to recover the correct tree structure and their prediction power are demonstrated through simulation. They are also applied to two real data examples where the covariates are time-varying. 

An \texttt{R} script for implementing LTRCIT and LTRCART is available at \url{http://people.stern.nyu.edu/jsimonof/survivaltree}. A corresponding \texttt{R} package \texttt{LTRCtrees} has been submitted to CRAN. 

\newpage
\bibliographystyle{apalike}
\bibliography{myreferences}
\end{document}